\documentclass[conference, onecolumn]{IEEEtran}
\usepackage{fullpage}

\usepackage[T1]{fontenc} 
\usepackage{amsmath,amssymb,latexsym,float,amsthm}
\usepackage[cmintegrals]{newtxmath}
\usepackage{tabularx}
\usepackage{bm} 
\usepackage{float}
\usepackage{graphicx, subfigure,epstopdf}
\usepackage{epsfig}
\usepackage{multirow}
\usepackage{url}
\usepackage{algorithm}
\usepackage{algpseudocode, color,xcolor}
\usepackage{comment}
\usepackage{lineno}
\usepackage{caption}
\usepackage{array}

\newtheorem{thm}{Theorem}
\newtheorem{lem}{Lemma}

\newtheorem{prop}{Proposition}

\allowdisplaybreaks

\newcommand{\R}{\mathbb{R}}
\newcommand{\C}{\mathbb{C}}

\renewcommand{\Re}[1]{\operatorname{Re}\left\{#1\right\}}

\newcommand{\E}{\mathbb{E}}

\renewcommand{\d}[1]{d#1}

\newcommand{\vct}[1]{\boldsymbol{#1}}
\newcommand{\mtx}[1]{\boldsymbol{#1}}

\newcommand{\<}{\langle}
\renewcommand{\>}{\rangle}

\newcommand{\vctr}{\operatorname{vec}}

\newcommand{\set}[1]{\mathcal{#1}}

\newcommand{\vc}{\vct{c}}

\newcommand{\ve}{\vct{e}}
\newcommand{\vf}{\vct{f}}

\newcommand{\vh}{\vct{h}}

\newcommand{\vp}{\vct{p}}
\newcommand{\vq}{\vct{q}}
\newcommand{\vr}{\vct{r}}
\newcommand{\vs}{\vct{s}}

\newcommand{\vu}{\vct{u}}
\newcommand{\vw}{\vct{w}}
\newcommand{\vx}{\vct{x}}
\newcommand{\vy}{\vct{y}}
\newcommand{\vz}{\vct{z}}
\newcommand{\mC}{\mtx{C}}

\newcommand{\mF}{\mtx{F}}
\newcommand{\mG}{\mtx{G}}
\newcommand{\mH}{\mtx{H}}
\newcommand{\mI}{\mtx{I}}

\newcommand{\mR}{\mtx{R}}

\newcommand{\mX}{\mtx{X}}

\newcommand{\mZ}{\mtx{Z}}
\newcommand{\setA}{\set{A}}

\newcommand{\setC}{\set{C}}

\newcommand{\setG}{\set{G}}
\newcommand{\setH}{\set{H}}
\newcommand{\setI}{\set{I}}

\newcommand{\setN}{\set{N}}
\newcommand{\setO}{\set{O}}
\newcommand{\setP}{\set{P}}

\newcommand{\setW}{\set{W}}
\newcommand{\setX}{\set{X}}

\newcommand{\setZ}{\set{Z}}

\newcommand{\PP}{\mathbb{P}}

\def\tf{\tilde{F}}
\def\gtf{\nabla \tilde{F}}

\begin{document}

\title{Blind Deconvolution Demixing using Modulated Inputs}

\author{\IEEEauthorblockN{Humera Hameed}
\IEEEauthorblockA{Department of Electrical Engineering\\
Information Technology University of Punjab\\
Lahore, Pakistan 54700\\
Email: phdee17001@itu.edu.pk}
\and
\IEEEauthorblockN{Ali Ahmed}
\IEEEauthorblockA{Department of Computer Science\\
Information Technology University of Punjab\\
Lahore, Pakistan 54700\\
Email: ali.ahmed@itu.edu.pk}}

\maketitle

\begin{abstract}
 This paper focuses on solving  a challenging problem of blind deconvolution demixing involving modulated inputs. Specifically, multiple input signals $s_n(t)$, each bandlimited to $B$ Hz, are modulated with known random sequences $r_n(t)$ that alter at rate $Q$. Each modulated signal is convolved with a different M tap channel of impulse response $h_n(t)$,  and the outputs of each channel are added at a common receiver to give the observed signal $y(t)=\sum_{n=1}^N (r_n(t)\odot s_n(t))\circledast h_n(t)$, where $\odot$ is the point wise multiplication, and $\circledast$ is circular convolution. Given this observed signal $y(t)$, we are concerned with recovering $s_n(t)$ and $h_n(t)$. We employ deterministic subspace assumption for the input signal $s_n(t)$ and keep the channel impulse response $h_n(t)$ arbitrary. We show that if modulating sequence is altered at a rate $Q \geq N^2 (B+M)$ and sample complexity bound is obeyed then all the signals and the channels, $\{s_n(t),h_n(t)\}_{n=1}^N$, can be estimated from the observed mixture $y(t)$ using gradient descent algorithm. We have performed extensive simulations that show the robustness of our algorithm and used phase transitions to numerically investigate the theoretical guarantees provided by our algorithm.

\end{abstract}




\section{Introduction}\label{sec:introduction}
This paper focuses on solving the joint problems of blind deconvolution  and demixing when multiple modulated  inputs are involved. Specifically, multiple input signals $\{s_n(t)\}_{n=1}^N$, $t \in [0,1)$,  are modulated with random binary waveforms (spread spectrum) $\{r_n(t)\}_{n=1}^N$  similar to \cite{ahmed2018ModBD} and passed through different channels of impulse response $h_n(t)= \sum_{m=1}^M h_n[m] \delta (t-t_m)$ with $t_m \in T_Q = \{ 0,\frac{1}{Q},..., 1-\frac{1}{Q} \}$. The outputs of all the channels are added at a common base station to generate the observed signal 
$$y(t)= \sum_{n=1}^N (s_n(t)\odot r_n(t))\circledast h_n(t),$$
Given this observed signal $y(t)$, we aim to recover $s_n(t)$ and $h_n(t)$ when neither is known is called blind deconvolution demixing (BDD). This problem arises in multiple applications including image processing \cite{Romberg2015app,Shamir2011app}, underwater acoustics \cite{Mansoor2006app}, and wireless communication \cite{Wang1998app}.

In this paper we focus on sixth generation (6G) wireless communication when channel state information (CSI) is not available.
In 6G, using orthogonal multiple access (OMA) in downlink (DL), base stations (BS) can generate different orthogonal spreading sequences for different user equipment (UE) to make the signal separable. However, in uplink (UL), OMA is not applicable due to the multiple UE generating their own spreading sequences and synchronization issues \cite{Wang2006OMA}. 
NOMA is the key solution to the problem of multiple access in 6G. NOMA has basically two types: power domain NOMA, and code domain NOMA \cite{Dai2015TypeNoma}. 
In power domain NOMA, different powers are allocated to different UEs depending upon their channel states \cite{Liu2017NonorthogonalMA, Imari2014NonorthogonalMA}. In code domain NOMA, all the users use the same time-frequency resources, and different spread spectrum to make the signals separable at the receiver. Our proposed algorithms has application in code domain NOMA and we will use spread spectrum to easily separate the signals as shown in Figure \ref{fig:SystemModel}. In practice, CSI requires channel training which consumes extra resources. For fast-moving objects, channel is continuously changing, and even with channel training, CSI will not be accurate which could cause error in UEs signal recovery. To overcome these drawbacks, we present a gradient descent algorithm for blind deconvolution and demixing which will avoid expensive channel training procedures. 
  
  A convex approach to blind deconvolution demixing (BDD) was proposed in \cite{li2016rapid}. The authors employed random coding assumptions and gave sub-optimal theoretical bounds that were later improved to near-optimal in \cite{Stoeger2017BDD}. However, the approach in \cite{Stoeger2017BDD} was computationally expensive as it involved the lifting of variables. In \cite{li2018rapid}, author proposed a computationally efficient non-convex approach for BDD using gradient descent algorithm for random coding subspaces. After that, \cite{Strohmer2018Demix} proposed two iterative hard thresholding-based algorithms for blind demixing of low-rank matrices, but their theoretical analysis are not applicable to BDD. Authors in \cite{Dong2019BDD} proposed a Riemannian optimization based solution to BDD but they also employed i.i.d Gaussian assumption for the rows of the coding matrix in theoretical analysis. In \cite{ahmed2018ModBD}, author assumed that the channel impulse response $h_n(t)$ is arbitrary and employ deterministic subspace assumption for the input signal $s_n(t)$ for blind deconvolution problem. We use the same deterministic subsapce assumptions for BDD that is unlike previous convex and non convex approaches \cite{ahmed2012blind,ahmed2015convex,ahmed2016leveraging,ahmed2018convex,li2016rapid,li2018rapid,lee2017spectral} which assumed random subspaces.

\begin{figure}
	\centering
	\begin{tabular}{cc}
		\includegraphics[scale = 0.5, trim = 10.3cm 9cm 0cm 0cm,clip]{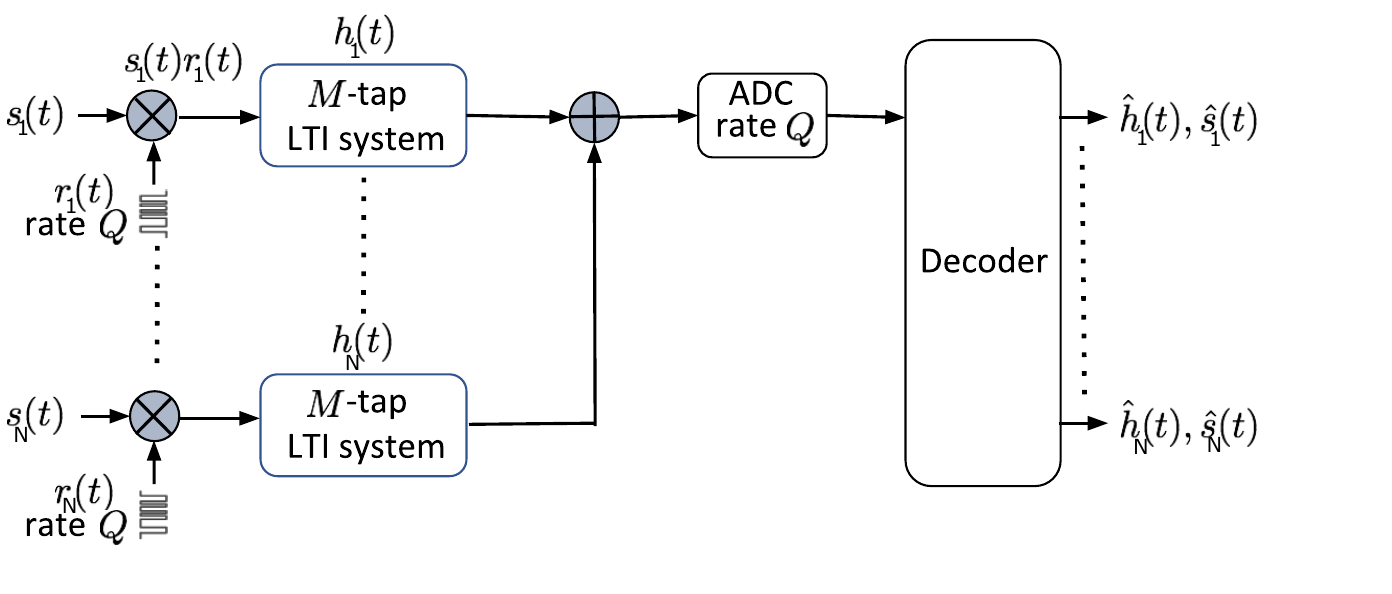}
	\end{tabular}
	\caption{\small\sl Analog implementation of uplink NOMA using single carrier OFDM for real time protection against channel interference.  Different user equipment's (UE) are transmitting continuous time signal $s_n(t)$ which is bandlimited to $B$ Hz. At each transmitter, random binary waveform $r_n(t)$, alternating at a rate $Q$, is used to modulate $s_n(t)$. The modulated signals are passed through different unknown LTI channels having an $M$-tap impulse response $h_n(t)$ and received at the common base station (BS). The received signal, $y(t)$, is a mixture of the convolution of the transmitted signals and the channels through which the signals passed. At a sampling rate $Q$, the received signal is sampled by ADC where $Q \gtrsim N^2 (B+M)$ (scale with coherences), and recover unknown signals $\{s_n(t)\}_{n=1}^N$, and channels $\{h_n(t)\}_{n=1}^N$ using algorithm \ref{algo:gradient-descent}.}
	\label{fig:SystemModel}
\end{figure}

\section{Problem Setup} \label{sec:Problem_setup}
The received signal $y(t)$ sampled at rate $L \geq Q$, in Fourier domain, can be given as
\begin{align}\label{eq:model}
\hat{\vy} = \sqrt{L} \sum_{n=1}^N \mF_Q (\vr_n \odot \vs_{n}) \odot \mF_M \vh_n,
\end{align}
where $\mF_Q$, and $\mF_M$ are formed by selecting first $Q$, and $M$ columns of a normalized $L \times L$ DFT matrix and $\vh_{n0} \in \C^M$, $\vx_{n0} \in \C^K$ are the ground truths, $\vs_{n0}=\mC_n \vx_{n0} \in \C^Q$ for $\mC_n \in \C^{Q \times K}$. Similarly, we can write the noisy measurements in Fourier domain $\hat{\vy} \in \C^L$ as
$$\hat{\vy} = \mF\vy = \sqrt{L} \sum_{n=1}^N [\mF_Q (\vr_n \odot \vs_{n0}) \odot \mF_M \vh_{n0}] + \hat{\ve}$$

\begin{align}\label{eq:measurements}
\hat{\vy} = \sqrt{L} \sum_{n=1}^N [\mF_Q \mR_n \mC_n \vx_{n0} \odot \mF_M \vh_{n0}] + \hat{\ve}
\end{align}
where $\mR_n = diag(\vr_n) \in \C^Q$  and $\hat{\ve} \in \mathbb{C}^L$ is the additive noise. 

The expression for the observed vector \eqref{eq:measurements}, although non linear in $(\vh_{n0},\vx_{n0})$, is linear in the rank-1 outer product $\vh_{n0}\bar{\vx}_{n0}^*$. Let $\vf^*_\ell \in \C^M$ be the $\ell$th row of $\mF_M$, and $\hat{\vc}_{n\ell}^* \in \C^{K}$ be the $\ell$th row of the  $L \times K$ matrix $\sqrt{L}(\mF_Q\mR_n\mC_n)$. Then, $\ell$th entry of $\hat{\vy}$ in \eqref{eq:measurements} can be easily expressed as follows,
\begin{align}\label{eq:entrywise-measurements}
\hat{y}[\ell] = \sum_{n=1}^N \vf_\ell^*\vh_{n0}\bar{\vx}_{n0}^*\hat{\vc}_{n\ell}+\hat{e}[\ell] = \sum_{n=1}^N \< \vf_\ell\hat{\vc}_{n\ell}^*,\vh_{n0}\bar{\vx}_{n0}^*\>+\hat{e}[\ell],
\end{align}
Linearity of $\hat{y}$ in the outer product $\vh_{n0}\bar{\vx}_{n0}^*$ is quite evident from  \eqref{eq:entrywise-measurements}. A linear map $\setA_n: \C^{M \times K} \rightarrow \C^{L}$ is defined that maps $\vh_{n0}\bar{\vx}_{n0}^*$ to  $\hat{\vy}$. $\setA_n$ acts on a rank-1 matrix $\vh_{n0}\bar{\vx}_{n0}^*$ to return the following output
\begin{align}\label{eq:linear-map}
&\setA_n(\vh_{n0}\bar{\vx}_{n0}^*) := \{\vf_\ell^*\vh_{n0} \bar{\vx}_{n0}^*\hat{\vc}_{n\ell}\}_{\ell}, \ (\ell)  \in [L], \notag\\
&\text{and therefore,} \ \hat{\vy} =  \sum_{n=1}^N \setA_n(\vh_{n0}\bar{\vx}_{n0}^*) + \hat{\ve},
\end{align}
where $\setA_n$ is used to express \eqref{eq:entrywise-measurements} in a more compact manner. For $N$ concatenated vectors, $\vh_0=\vctr[\vh_{n0}], \vx = \vctr[\vx_{n0}]$, block diagonal matrix is shown as $\mZ(\vh_0,\vx_0) = (\mZ_n(\vh_{n0},\vx_{n0}))^{\otimes N}$ where $\mZ_n(\vh_n,\vx_n)=\vh_n \vx_n^*$ is the outer product of two vectors. So,
We can also write $\hat{\vy}$ as
\begin{align}\label{eq:linear-map-whole}
 \hat{\vy} =  \setA (\mZ(\vh_0,\vx_0)) + \hat{\ve},
\end{align}
where $\setA (\mZ(\vh_0,\vx_0)) = \sum_{n=1}^N \setA_n (\vh_{n0}\vx_{n0}^*).$  

\subsection{Coherence Parameters}\label{sec:coherence_parameters}
For arbitrary vectors $\vh := \vctr([\vh_n])\in \C^{MN}$, $\vx := \vctr([\vx_n])\in \C^{KN}$, where $[\vx_n] \in \C^{K}$, $[\vh_n] \in \C^{M}$, we define coherences 
\begin{align}\label{eq:muh-nux}
 \mu^2 : = \max_{1 \leq n \leq N} L  \frac{\|\mF_M\vh_{n0}\|_\infty^2}{\|\vh_{n0}\|_2^2}, \text{and}\  \nu^2 := \max_{1 \leq n \leq N} Q \frac{\|\mC_{n}\vx_{n0}\|_\infty^2 }{\|\vx_{n0}\|_2^2},  
 \nu^2_{\max} := Q \|\mC_{n}^{\otimes N} \|_\infty^2
\end{align} 

In general, we assume that $ \|\vh_{n0}\|_2^2= d_{n0}$, $ \|\vx_{n0} \|_2^2 = d_{n0}$ and $ \sum_{n=1}^N \|\vh_{n0}\|_2^2 \|\vx_{n0}\|_2^2 = \sum_{n=1}^N d^2_{n0} = d^2_0$. Coherence parameter $\mu^2$ attains a minimum value if a fixed norm vector $\vh_n$ has disperse spectrum and vice versa, so $1 \leq \mu^2 \leq L$. 

Dispersion of the signals $\vs_n = \mC_n\vx_n$ in time domain is expressed by $\nu^2$, similar term appeared in \cite{ahmed2018ModBD}. Let $\vx_n$ be highly aligned with few of $\vc_{q,n}^*$ then $\nu^2 \leq Q$ where $\vc_{q,n}^*$ is the $q$-th row of $\mC_n$ for any $(q,n) \in [Q] \times [N]$. For well dispersed $\mC_n \vx_n$, we have lower bound $ 1 \leq \nu^2$, hence, $ 1 \leq \nu^2 \leq Q$.

\section{WIRTINGER GRADIENT DESCENT ALGORITHM} \label{sec:gradient_descent}  
Given the observed signal $\hat{\vy}$, we formulate the problem of recovering the ground truth signals as a minimization problem. More specifically, we minimize the following loss function with respect to  $\{\vh_n,\vx_n\}_{n=1}^N$ using a regularized gradient descent algorithm:
\begin{align}\label{eq:tF-def}
\tf(\vh,\vx) := F(\vh,\vx)+G(\vh,\vx).
\end{align}
where $\vh= \text{vec}[\vh_n]$, $\vx= \text{vec}[\vx_n]$ for $n=1,2,...,N$. In equation (\ref{eq:tF-def}) the function $F(\vh,\vx)$ deals with measurement loss and is defined as
\begin{align}\label{eq:F-def}
F(\vh,\vx) := & \bigg \| \sum_{n=1}^N \setA_n(\vh_n\bar{\vx}_n^*-\vh_{n0}\bar{\vx}_{n0}^*)-\ve\bigg \|_2^2 \notag \\
 = & \bigg \| \sum_{n=1}^N \setA_n(\vh_n\bar{\vx}_n^*-\vh_{n0}\bar{\vx}_{n0}^*) \bigg\|_2^2 + 
\|\ve\|_2^2 - 
2 \text{Re}(\<\setA^*(\ve), \mZ(\vh,\vx)- \mZ(\vh_{0},\vx_{0} \>),
\end{align}
and $G(\vh,\vx)$ accounts for restricting the coherences $\mu_h^2$, $\nu_x^2$; and norms of $\{ \vh_n,\vx_n \}_{n=1}^N$ to within a vicinity of norms of the ground truth $\{\vh_{n0},\vx_{n0}\}_{n=1}^N$, and is defined below 
\begin{align}\label{eq:G-def}
G(\vh,\vx) := \sum_{n=1}^N G_n(\vh_n,\vx_n)
\end{align}
\begin{align}\label{eq:Gi-def}
 G_n(\vh_n,\vx_n) := \rho\Bigg[ G_0\left( \frac{\|\vh_n\|_2^2}{2d_n} \right)+ G_0\left(\frac{\|\vx_n\|_2^2}{2d_n} \right) 
 +\sum_{\ell=1}^L G_0\left( \frac{L |\vf_\ell^* \vh|^2}{8d_n \mu^2}\right) + \sum_{q=1}^Q G_0\left( \frac{Q |\vc_{q,n}^* \vx_n|^2}{8d_n \nu^2}\right)\Bigg], 
\end{align}
where $G_0 (z ) = \max\{z-1,0\}^2$. To prove Theorem \ref{thm:initialization}, we set $\rho \geq d^2+\|\ve\|_2^2$, $ 0.9 d_{0} \leq d \leq 1.1 d_{0}$, and $ 0.9 d_{0n} \leq d_n \leq 1.1 d_{0n}$. 
 In each iteration of the proposed regularized gradient descent algorithm \ref{algo:gradient-descent}, an  alternating minimization strategy is employed which minimizes  the loss function with respect to $\vh_n$, and $\vx_n$, for a particular value of n,  while keeping the rest fixed.
\begin{algorithm}[H]
	\caption{Regularized Wirtinger gradient descent with a step size $\eta$}
	\label{algo:gradient-descent}
	\begin{algorithmic}\small
	    \State \textbf{Input:} Obtain $\{\vu_n^0,\boldsymbol{v}_n^0\}_{n=1}^N$ via Algorithm \ref{algo:initialization} below.
		\For{$t = 1, \ldots$}
			\For{$n = 1, \ldots, N$}
		    \State $ \vu_n^t \gets  \vu_n^{t-1} - \eta \nabla \tilde{F}_{\vh_n} (\vu_n^{t-1},\boldsymbol{v}_n^{t-1} )$
			\State $ \boldsymbol{v}_n^t \gets  \boldsymbol{v}_n^{t-1} - \eta \nabla \tilde{F}_{\vx_n} (\vu_n^{t-1},\boldsymbol{v}_n^{t-1} )$
			\EndFor
		\EndFor
		\State \textbf{Output:} $\{\hat{\vh}_{n0},\hat{\vx}_{n0}\}_{n=1}^N$
	\end{algorithmic}
\end{algorithm}
 
 The gradient descent updates require the Wirtinger gradients of the loss function with respect to $\vh_n$, and $\vx_n$ which are defined as 
\footnote{The Wirtinger gradient for a complex function $f(\vz)$ is defined as 
$ \frac{\partial f}{\partial \bar{\vz}} = \frac{1}{2}\left( \frac{\partial f}{\partial \vx} + \iota \frac{\partial f}{\partial \vy}  \right),$where $\vz = \vx + \iota \vy \in \C^L$, and $\vx, \vy \in \R^L$.}

\begin{align}\label{eq:gFh-gFx-def}
    \nabla \tilde{F}_{\vh} = \begin{bmatrix}
           \nabla \tilde{F}_{\vh_1} \\
           \nabla \tilde{F}_{\vh_2} \\
           \vdots \\
           \nabla\tilde{F}_{\vh_N}
         \end{bmatrix},
             \nabla \tilde{F}_{\vx} = \begin{bmatrix}
           \nabla \tilde{F}_{\vx_1} \\
           \nabla \tilde{F}_{\vx_2} \\
           \vdots \\
           \nabla \tilde{F}_{\vx_N}
         \end{bmatrix} 
         \text{where,} 
\nabla \tilde{F}_{\vh_n} := \frac{\partial \tilde{F}}{\partial \bar{\vh}_n} = \frac{\overline{\partial \tilde{F}}}{\partial {\vh_n}}, \ \text{and} \ \nabla \tilde{F}_{\vx_n} := \frac{\partial \tilde{F}}{\partial \bar{\vx}_n} = \frac{\overline{\partial \tilde{F}}}{\partial {\vx_n}}.
\end{align}

By linearity, $\nabla \tilde{F}_{\vh_n} = \nabla F_{\vh_n} + \nabla G_{\vh_n}$, and similarly, $\nabla \tilde{F}_{\vx_n} = \nabla F_{\vx_n} + \nabla G_{\vx_n}$ where the gradients $\nabla \tilde{F}_{\vh}$, and $\nabla \tilde{F}_{\vx}$ are defined in \eqref{eq:gFh-gFx-def}. As $F$, and $G$ are defined in \eqref{eq:F-def}, and \eqref{eq:G-def}, the gradients w.r.t. $\vh_n$, and $\vx_n$ can be written as 
\begin{align}\label{eq:gF-def}
\nabla F_{\vh_n}  &= \setA_n^*(\setA(\mZ(\vh,\vx)-\mZ(\vh_0,\vx_0))-\ve) \vx_n, \notag\\
 \nabla F_{\vx_n} &= [\setA_n^*(\setA(\mZ(\vh,\vx)-\mZ(\vh_0,\vx_0))-\ve)]^*\vh_n,
\end{align}
and
\begin{align}\label{eq:ghG-def}
\nabla G_{\vx_n} = \frac{\rho}{2d_n}\Bigg[ G_0^\prime \left( \frac{\|\vx_n\|_2^2}{2d_n}\right) \vx_n 
 + \frac{Q}{4\nu^2} \sum_{q=1}^Q G_0^\prime \left( \frac{Q |\vc_{q,n}^*\vx_n|^2}{8d_n \nu^2}\right) \vc_{q,n}\vc_{q,n}^* \vx_n \Bigg], \notag \\
\nabla G_{\vh_n} = \frac{\rho}{2d_n}\Bigg[ G_0^\prime \left( \frac{\|\vh_n\|_2^2}{2d_n}\right) \vh_n +
 \frac{L}{4\mu^2} \sum_{\ell=1}^L G_0^\prime \left( \frac{L |\vf_{\ell}^*\vh_n|^2}{8d_n \mu^2}\right) \vf_{\ell}\vf_{\ell}^* \vh_n\Bigg].
\end{align}

For Algorithm \ref{algo:gradient-descent} we obtain a good initialization $\{\vu_n^0,\boldsymbol{v}_n^0\}_{n=1}^N$ by using Algorithm \ref{algo:initialization}. In words, the projections of the left and right singular vectors of $\setA_n^*(\hat{\vy}):= \sum_{l=1}^L \hat{\vy}(l) \vf_l \hat{\vc}_{nl}^*$, corresponding to leading singular value, into the set of incoherent vectors  serve as the initializers 
$\{\vu_n^0,\boldsymbol{v}_n^0\}_{n=1}^N$ for Algorithm \ref{algo:gradient-descent}.
\begin{algorithm}[H]
	\caption{Initialization of unknown signals}\small
	\label{algo:initialization}
	\begin{algorithmic}
		\State \textbf{Input:} $\hat{\vy}$, $\setA$ 
		\For{$n = 1, \ldots, N$}
		\State Compute left and right singular vectors of $\setA_n^*(\hat{\vy})$,   $\hat{\vh}_{n0}$, and $\hat{\vx}_{n0}$, respectively corresponding to the leading singular value $d_n$. 
		\State Solve the following optimization programs 
		\State $\vu_n^0 \gets \underset{\vh_n}{\text{argmin}} \ \|\vh_n - \sqrt{d_n} \hat{\vh}_{n0}\|_2, \ \text{subject to} \  \sqrt{L}\|\mF_M\vh_n\|_\infty \leq 2 \sqrt{d_n} \mu,$\ \text{and} \ 
		\State $\boldsymbol{v}_n^0 \gets \underset{\vx_n}{\text{argmin}} \ \|\vx_n - \sqrt{d_n} \hat{\vx}_{n0}\|_2, \ \text{subject to} \ \sqrt{Q}\|\mC_n\vx_n\|_\infty \leq 2 \sqrt{d_n} \nu.$
		\EndFor
		\State \textbf{Output:} $\{\vu_n^0,\boldsymbol{v}_n^0\}_{n=1}^N$. 
	\end{algorithmic}
\end{algorithm}

\subsection{Neighborhood Sets}\label{sec:neighborhood_sets} 
Stable and robust recovery of the ground truth is made possible by ensuring that the iterates of the gradient descent algorithm \ref{algo:gradient-descent} remain within a region of incoherence and in close vicinity to the ground truth. To formalize these notions, we define the following sets of neighboring points of $\{(\vh_n,\vx_n)\}_{n=1}^N$

\begin{align}
&\setN_{d_0} := \{\{(\vh_n,\vx_n)\}_{n=1}^N | \|\vh_n\|_2 \leq 2\sqrt{d_{n0}}, \ \|\vx_n\|_2\leq 2 \sqrt{d_{n0}} \},\label{eq:setNd-def}\\
&\setN_\mu := \{ \{(\vh_n,\vx_n)\}_{n=1}^N | \sqrt{L}\|\mF_M\vh_n\|_\infty \leq 4\mu\sqrt{d_{n0}} \},\label{eq:setNmu-def}\\
&\setN_\nu : = \{\{(\vh_n,\vx_n)\}_{n=1}^N|  \sqrt{Q}\|\mC_n\vx_n\|_\infty \leq 4\nu\sqrt{d_{n0}} \},\label{eq:setNnu-def}\\
&\setN_{\varepsilon} := \{\{(\vh_n,\vx_n)\}_{n=1}^N | \|\vh_n\bar{\vx}_n^*-\vh_{n0}\bar{\vx}_{n0}^*\|_{F} \leq \varepsilon d_{n0} \}\label{eq:setNe-def}.
\end{align}

\section{Main Results}
Our main result on blind deconvolution demixing using deterministic subspaces \eqref{eq:measurements} is given below.

\begin{thm}\label{thm:convergence}
	 Fix $0 < \varepsilon \leq 1/15$. Let $\mC_n \in \R^{Q \times K}$ be a tall basis matrix, and set $\vs_{n0} = \mC_n\vx_{n0}$; and $\vx_{n0} \in \C^K$, $\vh_{n0} \in \C^M$ be arbitrary vectors for every $n = 1,2,3, \ldots, N$. Let the coherence parameters of $\mC_n$, $(\vh_{n0},\vx_{n0})$ be as given in \eqref{eq:muh-nux}, and $\kappa = \frac{\max d_{n0}}{\min d_{n0}}$.  Let each $[\vr_n]$ be an independent $Q$-length vector with standard iid Rademacher entries. The $L$-point noisy sum of circular convolutions of the unknown channels $\vh_{n0}$ with random sign vectors $\vr_n\odot\vs_{n0}$ for all $n = 1,2,3, \ldots, N$, is observed as defined in \eqref{eq:measurements}. Let the initial guess $\{(\vu_{n0},\boldsymbol{v}_{n0})\}_{n=1}^N$ of $\{(\vh_{n0}, \vx_{n0})\}_{n=1}^N$ belongs to $\tfrac{1}{\sqrt{3}}\setN_{d_0} \cap \tfrac{1}{\sqrt{3}}\setN_{\mu}\cap \tfrac{1}{\sqrt{3}}\setN_{\nu} \cap \setN_{\frac{2\varepsilon}{5\sqrt{N}\kappa}},$ and that
	 \begin{align}\label{eq:sample-complexity-main-thm}
	 Q \geq c\frac{\kappa^2 N^2}{\xi^2 \delta^2_t} ( \mu^2 \nu^2_{\max} K + \nu^2 M)\log^4(LN),
	 \end{align}
	with $L \geq Q$, then Algorithm \ref{algo:gradient-descent} will create a sequence $\{(\vu_{n}^t,\boldsymbol{v}_n^t)\}_{n=1}^N \in \setN_{d_0}\cap \setN_{\mu} \cap \setN_{\nu} \cap \setN_{\varepsilon}$, which converges linearly to $\{(\vh_{n0},\vx_{n0})\}_{n=1}^N$ with probability at least 
	\begin{align}\label{eq:probability-main-thm}
	1-2\exp\left(-c \xi^2 \delta_t^2QN/\mu^2 \nu^2 \kappa^2 \right),
	\end{align}
	and there holds 
	\begin{align}\label{eq:stable-recovery-bound}
	\|\mZ(\vu^t,\boldsymbol{v}^t)-\mZ(\vh_0,\vx_0)\|_F \leq 
	 \frac{\epsilon d_0}{\sqrt{2N\kappa^2}} (1-\eta\omega)^{t/2}
	+ 60 \sqrt{N} \|\setA^*(\ve)\|_{2\rightarrow 2},
	\end{align}
	where $\eta$ represents the fixed step size, $\delta_t = \|\mZ(\vu^t,\boldsymbol{v}^t) - \mZ(\vh_{0},\vx_{0})\|_F/d_0$, and $\omega > 0$. For fixed $\alpha \geq 1$ and additive noise $\ve \sim \text{Normal}(\mathbf{0},\frac{\sigma^2 d_0^2}{2L}\mI_{L}) + \iota \text{Normal}(\mathbf{0},\frac{\sigma^2 d_0^2}{2L}\mI_{L})$, $\|\setA^*(\ve)\|_{2\rightarrow 2} \leq \frac{2\varepsilon}{50N \kappa} d_0$ with probability at least $1-\setO(L^{-\alpha})$ whenever 
	\begin{align}\label{eq:sample-complexity-LN}
	L \geq  c_\alpha\frac{(\mu_h^2  + \sigma^2)}{\varepsilon^2} \kappa^4 N^2 \max (M,K\log(L))\log(L).
	\end{align}
\end{thm}
 
\begin{thm}\label{thm:initialization}
	 Algorithm \ref{algo:initialization} provides the initialization such that $
	\{(\vu_{n0},\boldsymbol{v}_{n0})\}_{n=1}^{N} \in\frac{1}{\sqrt{3}}\setN_{d_0} \cap \frac{1}{\sqrt{3}}\setN_{\mu}  \cap \frac{1}{\sqrt{3}}\setN_{\nu}\cap \setN_{\frac{2\varepsilon}{5\sqrt{N}\kappa}},$ and $0.9 d_{n0} \leq d_n \leq 1.1 d_{n0}$, $0.9 d_0 \leq d \leq 1.1 d_0$ with probability at least $1-2\exp\left(-c\varepsilon^2 \delta_t^2Q
	/\mu^2\nu^2 \kappa^4  \right)$ whenever 
	\[
	Q \geq c \frac{ \kappa^4 N^2}{\varepsilon^2 \delta_t^2} \left( \mu^2 \nu_{\max}^2 K + \nu^2 M\right) \log^4(LN).
	\]
\end{thm}
Theorem \ref{thm:convergence} shows that good enough initialization for the Algorithm \ref{algo:gradient-descent} ensures linear convergence to the true solution.
Theorem \ref{thm:initialization} guarantees that the Algorithm \ref{algo:initialization} supplies the required initialization: $\{(\vu_{n0},\boldsymbol{v}_{n0})\}_{n=1}^{N} \in \frac{1}{\sqrt{3}}\setN_{d_0} \cap \frac{1}{\sqrt{3}}\setN_{\mu}  \cap \frac{1}{\sqrt{3}}\setN_{\nu}\cap \setN_{\frac{2\varepsilon}{5\sqrt{N}\kappa}}$.  
Dependence of sample complexity on $\frac{1}{\delta^2_t}$ in \eqref{eq:sample-complexity-main-thm} shows that we have approximate result for some fixed value of $\delta_t$. For exact recovery $\delta_t =0$ which means we need infinitely many samples, so we left this problem for future work. Proofs of Theorem \ref{thm:convergence}, and Theorem \ref{thm:initialization} are given in Supplementary Material Section D and E, respectively.
\begin{figure*}
	\centering
	\begin{tabular}{ccc}
		& \includegraphics[scale = 0.32, trim = 0.36cm 0cm 0cm 0cm,clip]{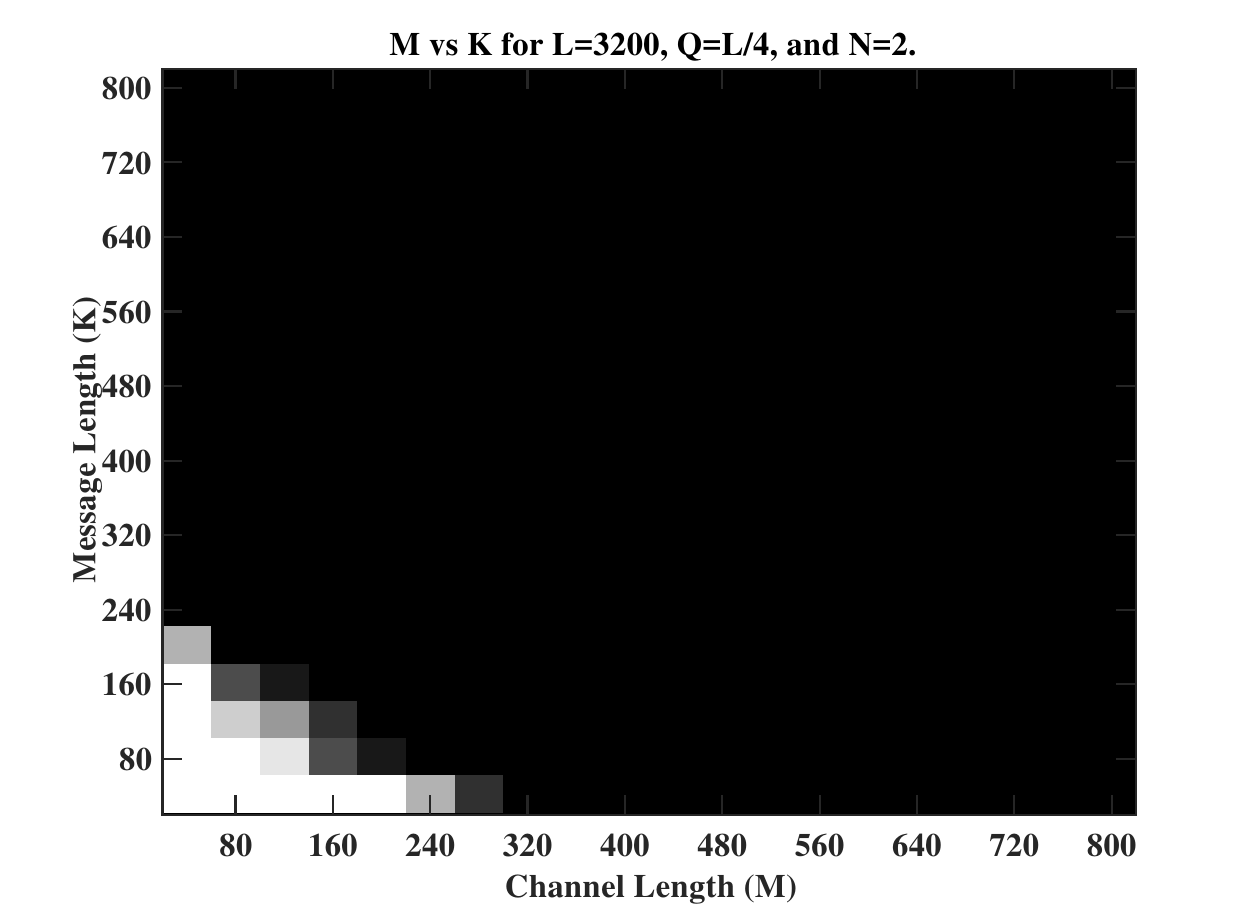}
		&\includegraphics[scale = 0.32]{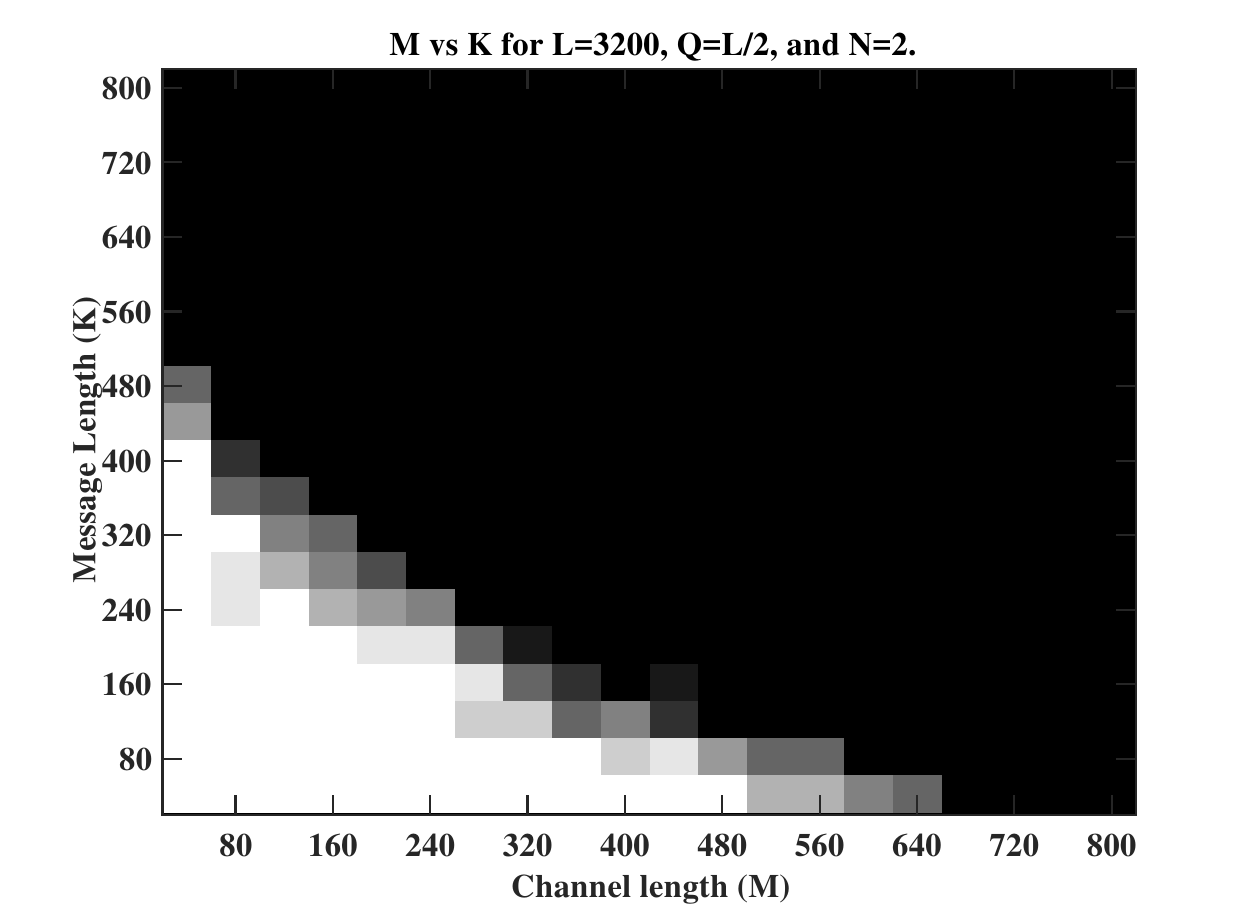}\\
		&\includegraphics[scale = 0.32, trim = 0cm 0cm 0.3cm 0cm,clip]{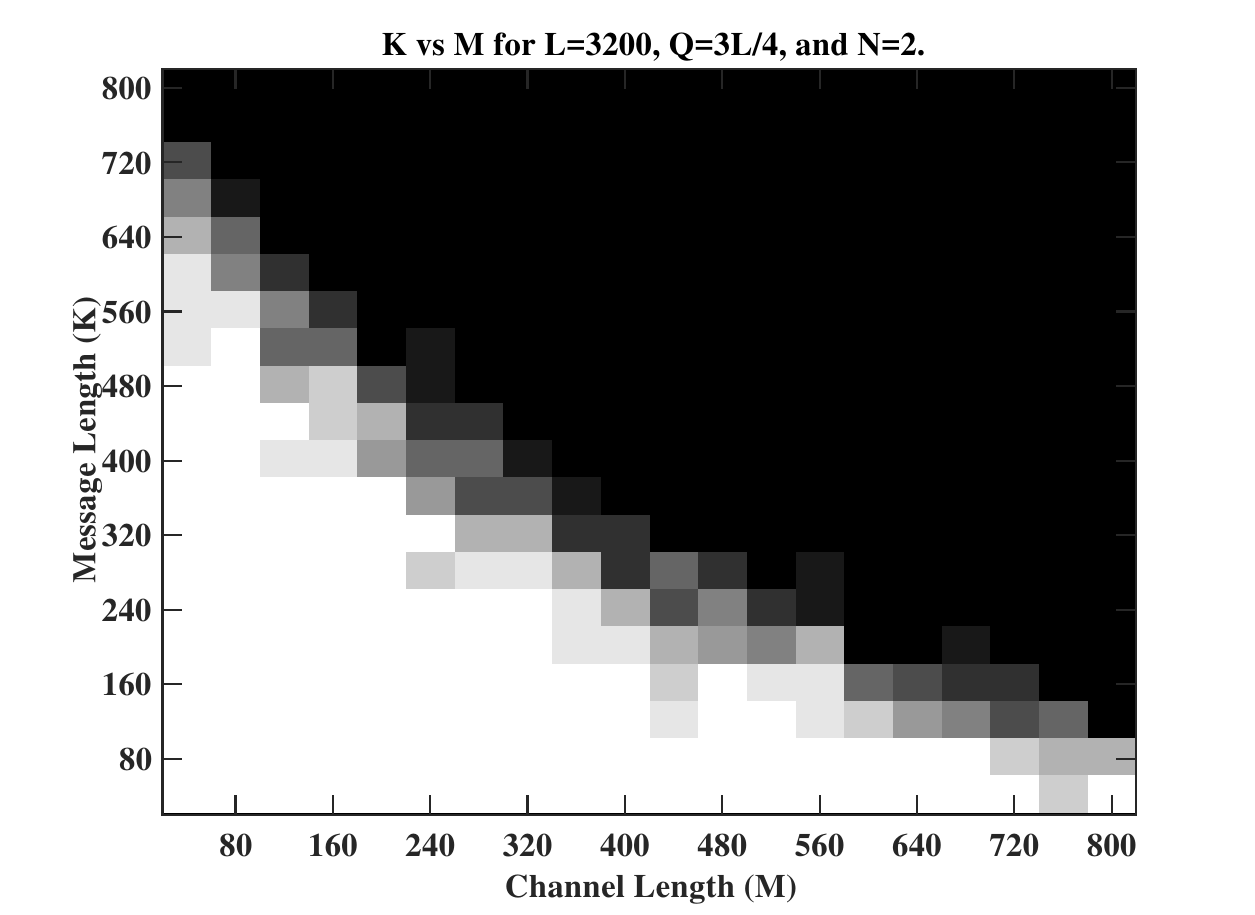}
		&\includegraphics[scale = 0.32, trim = 0.36cm 0cm 0cm 0cm,clip]{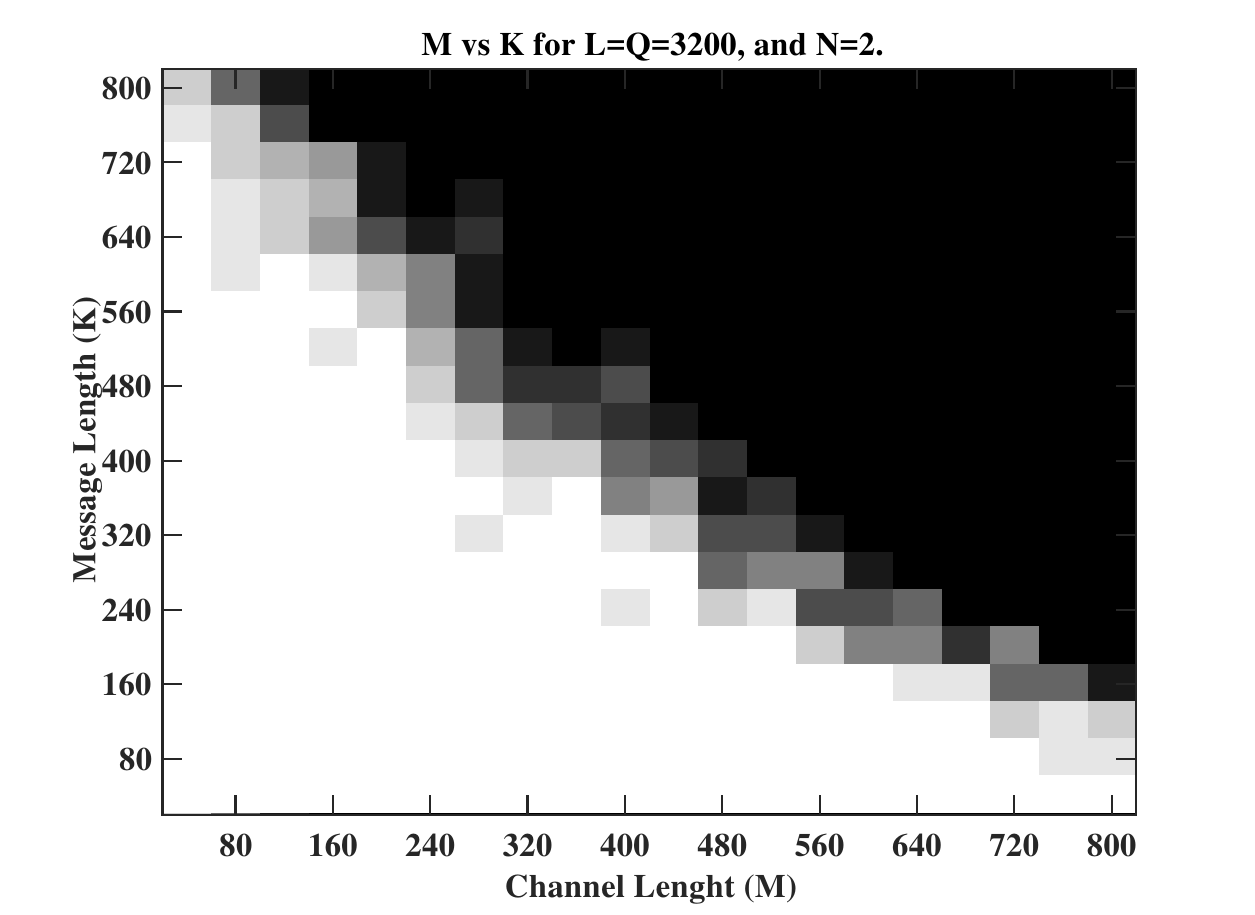}
        	\end{tabular}
	\caption{\small\sl For fixed $L$, and $N$, phase diagrams of $K$ vs. $M$ for different $Q$.  Phase transitions show that larger modulated inputs, larger $Q$, allow recovery with larger values of $K$, and $M$. }
	\label{fig:phase-transitions}
\end{figure*} 

\begin{figure*}[ht]
\centering
		\includegraphics[scale = 0.32, trim = 0.3cm 0cm 0cm 0cm,clip]{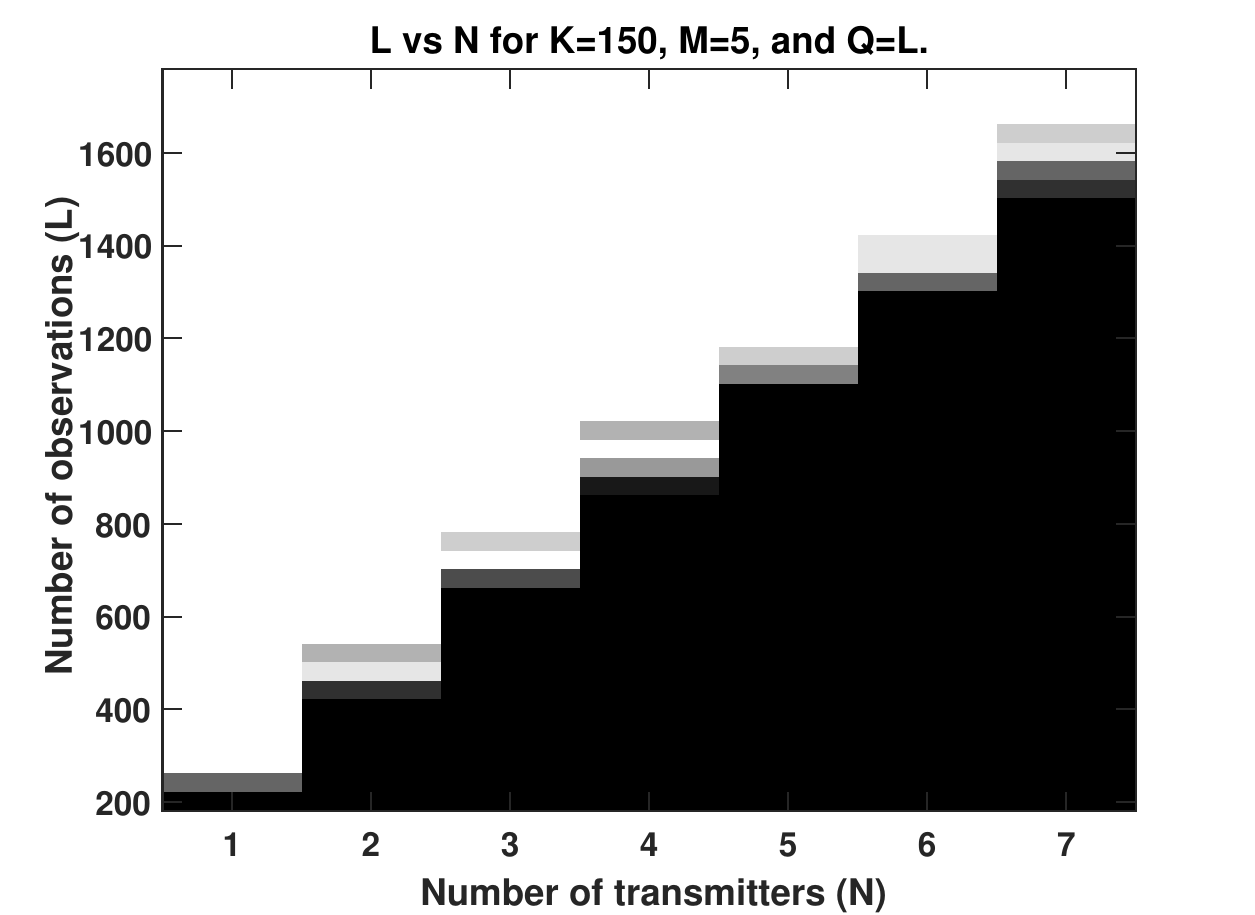}
       \caption{\small\sl Number of transmitters  vs. the number of observations.}
       \label{fig:NoOfTransmitters}
        \end{figure*}

\begin{figure*} [ht]  
\centering
		\includegraphics[scale = 0.32, trim = 0.3cm 0cm 0.3cm 0cm,clip]{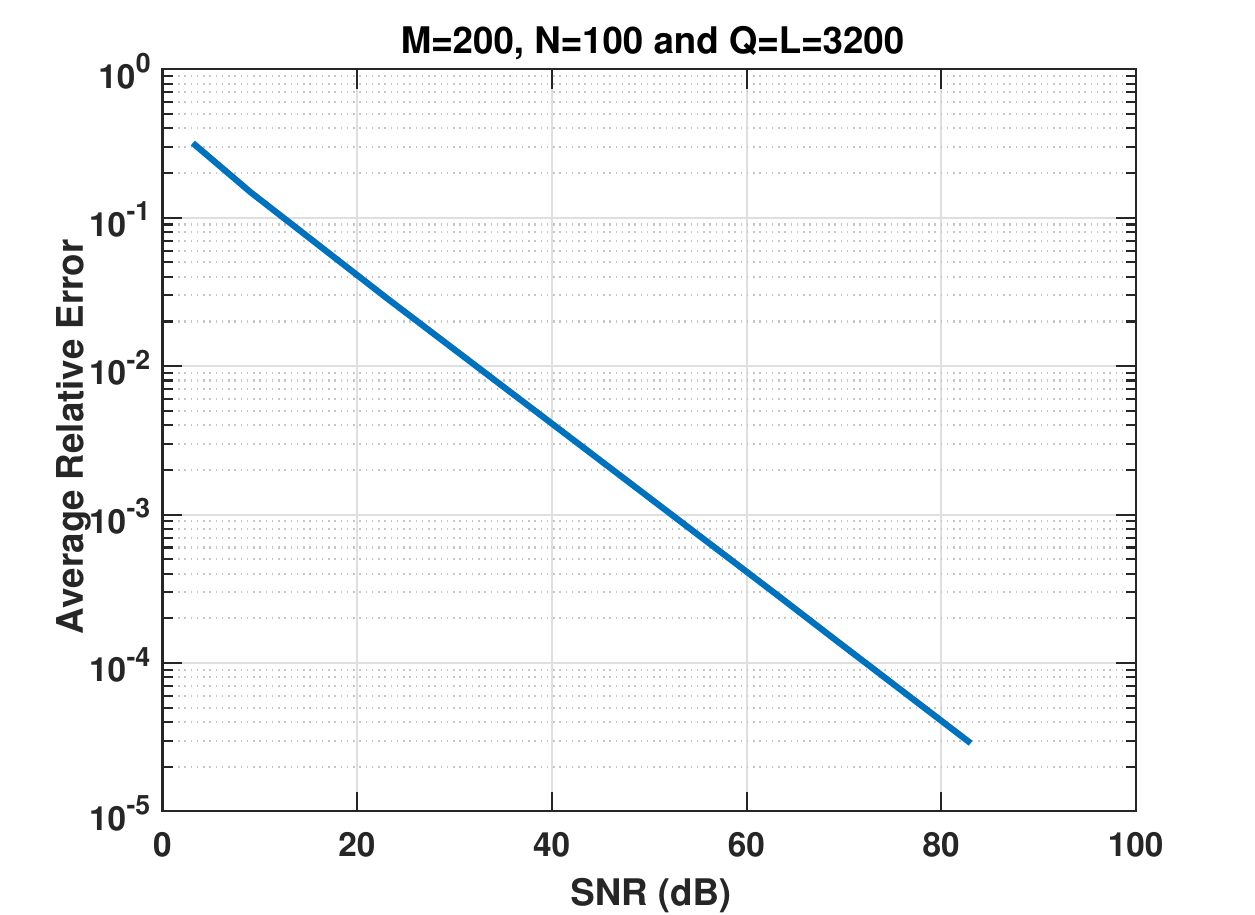}
        \caption{\small\sl SNR (dB) vs. average relative error for $N=2$.}
        \label{fig:snr}
        \end{figure*} 
        
\section{Numerical Simulations}\label{sec:numerics}
In this section, we provide extensive simulations, using phase transitions, to verify our sample complexity bounds. We have also done a stability analysis of our algorithm under a noisy environment. 

To numerically verify the bounds on the dimensions $Q$, $M$, $K$, and $L$ for fixed $N$ as given in \eqref{eq:sample-complexity-main-thm}, and \eqref{eq:sample-complexity-LN}, we generate phase transitions using Algorithm \ref{algo:gradient-descent}. Black region represents the probability of failure and white region represents the probability of success over $10$ independent experiments. We generate $\vh_{n0}$, and $\vx_{n0}$ as random Gaussian vectors, and set coding matrix $\mC_n$ as a subset of columns of the DCT matrix. Observed signal is obtained by using the model \eqref{eq:model}. Initialization is obtained via Algorithm \ref{algo:initialization} for Algorithm \ref{algo:gradient-descent} to converge. An experiment is labeled as success  if the relative error is below $10^{-2}$ where 
\begin{align}\label{eq:relative-error}
\text{Relative Error}: =\sqrt{ \frac{\sum_{n=1}^N \|\hat{\vh}_n \hat{\bar{\vx}}_n^*-\vh_{n0}\bar{\vx}_{n0}^*\|_F^2}{\sum_{n=1}^N\|\vh_{n0}\bar{\vx}_{n0}^*\|_F^2}}
\end{align} 
In Figure \ref{fig:phase-transitions}, phase diagrams for different lengths of modulated signals $Q$ are shown by fixing $L=3200$, and $N=2$. Message length $K$ and channel length $M$ are varying in each phase diagram by setting $Q$ to $L/4$, $L/2$, $3L/4$, and $L$ respectively. As $Q$ increases, the white region/probability of success increases. In first (top left) phase transition, when the number of measurements is $10$-times more than the number of unknowns, $L \approx 10N(K+M)$, almost always successful recovery occurs. From second to fourth phase diagram, this factor reduces to $5$, $3.5$ and $2$, respectively. 

Figure \ref{fig:NoOfTransmitters} shows the number of observations vs. number of transmitters $(N)$ assuming noise-free environment. We observe that number of observations required for exact recovery increases linearly with the increasing number of transmitters.

In Figure \ref{fig:snr}, the algorithm's performance at different noise levels is shown. We synthetically generate Gaussian noise vector $\hat{\ve}$ to generate observed noisy signal as in \eqref{eq:measurements}. The plot shows a deceasing trend for the relative error(log scale), averaged over ten independent experiments, with increasing $\text{SNR} := 10 \log_{10} \left(\|\vy\|_2^2/\|\ve\|_2^2\right)$.

\section{Conclusion and Future Work}
We have established fundamental limits and recovery guarantees for blind deconvolution demixing using modulated inputs that have implementation potential. The quadratic scaling in the number of components highlights the inherent difficulty of this problem compared to the multichannel case. In future, we can reduce quadratic to linear scaling through refined analysis and exploit additional structure like sparsity.

\appendix
In this appendix, we provide the proof of Theorems stated in the paper that depends on four key conditions similar to \cite{li2018rapid}. However, the main difference is in noise robustness and RIP conditions due to the realistic subspace assumptions i.e., spanned by spread code rather than idealistic random Gaussian matrices. Our work becomes exactly same as \cite{ahmed2018ModBD} for $N=1$ but we have solved a more difficult problem of blind deconvolution demixing. At different points, we will refer the reader to \cite{ahmed2018ModBD,li2018rapid} to avoid overlap and we try our best to use same notations for reader ease. Before stating these four conditions, we go through some preliminaries.
\section{Preliminaries} 
\subsection{Neighborhood Sets}\label{sec:neighborhood_sets} 
Stable and robust recovery of the ground truth is made possible by ensuring that the iterates of the gradient descent algorithm 1 remain within basin of attraction means a region of incoherence and in close vicinity to the ground truth. To formalize these notions, we define the following sets of neighboring points of $\{(\vh_n,\vx_n)\}_{n=1}^N$

\begin{align}
&\setN_{d_0} := \{\{(\vh_n,\vx_n)\}_{n=1}^N | \|\vh_n\|_2 \leq 2\sqrt{d_{n0}}, \ \|\vx_n\|_2\leq 2 \sqrt{d_{n0}} \},\label{eq:setNd-def}\\
&\setN_\mu := \{ \{(\vh_n,\vx_n)\}_{n=1}^N | \sqrt{L}\|\mF_M\vh_n\|_\infty \leq 4\mu\sqrt{d_{n0}} \},\label{eq:setNmu-def}\\
&\setN_\nu : = \{\{(\vh_n,\vx_n)\}_{n=1}^N|  \sqrt{Q}\|\mC_n\vx_n\|_\infty \leq 4\nu\sqrt{d_{n0}} \},\label{eq:setNnu-def}\\
&\setN_{\varepsilon} := \{\{(\vh_n,\vx_n)\}_{n=1}^N | \|\vh_n\bar{\vx}_n^*-\vh_{n0}\bar{\vx}_{n0}^*\|_{F} \leq \varepsilon d_{n0} \}\label{eq:setNe-def}.
\end{align}


Suppose that $\tilde{\vh}_n \perp \vh_{n0}$, and $\tilde{\vx}_n \perp \vx_{n0}$, so we can write $\vh_n$, and $\vx_n$ as $\vh_n = \alpha_{n1}\vh_{n0} + \tilde{\vh}_n$, and $\vx_n = \alpha_{n2} \vx_{n0} + \tilde{\vx}_n$, where 
$\alpha_{n1} = \frac{\vh_{n0}^*\vh_n}{d_{n0}},\ \text{and}\  \alpha_{n2} = \frac{\vx_{n0}^*\vx_n}{d_{n0}}.$ We define $\Delta\vh_n$, and $\Delta\vx_n$ as 

\begin{align}\label{eq:Deltah-Deltax}
\Delta \vh_n = \vh_n - \alpha_n \vh_{n0}, \ \text{and} \ \Delta \vx_n = \vx_n - \bar{\alpha}_n^{-1} \vx_{n0}.
\end{align}
To proof the Lemma \ref{lem:local-regularity-G}, we choose $\alpha_n$ as below:
\begin{align}
 \alpha_n(\vh_n,\vx_n) = \begin{cases}
(1-\delta_0)\alpha_{n1}, & \text{if}~ \|\vh_n\|_2 \geq \|\vx_n\|_2\\
\frac{1}{(1-\delta_0)\bar{\alpha}_{n2}}, & \text{if} ~ \|\vh_n\|_2 < \|\vx_n\|_2\notag
\end{cases}
\end{align}
where $\delta_0 := \tfrac{\delta}{10}$, and $\delta \leq \frac{\epsilon}{\sqrt{N}\kappa}$. We can write the difference of the outer product of the estimated and original variables in terms of the orthogonal decomposition of $\vh_n\vx_n^*$ as below.  
\begin{align}\label{eq:difference-expansion}
\vh_n\vx_n^*-\vh_{n0}\vx_{n0}^* = (\alpha_{n1}\bar{\alpha}_{n2}-1) \vh_{n0}\vx_{n0}^*+\bar{\alpha}_{n2} \tilde{\vh}_n\vx_{n0}^* + 
\alpha_{n1} \vh_{n0} \tilde{\vx}_n^* +\tilde{\vh}_n\tilde{\vx}_n^*.
\end{align}

To see the dependence of $\Delta \vh_n$, $\Delta \vx_n$, $\alpha_{n1}$, and $\alpha_{n2}$ on $\delta_n$ the lemma is stated below. 

\begin{lem}\label{lem:local-regulaity-Del-norm-bounds}
	Recall that $\|\vh_{n0}\|_2 = \|\vx_{n0}\|_2 = \sqrt{d_{n0}}$. If $\delta_n: = \frac{\|\vh_n\vx_n^*-\vh_{n0}\vx_{n0}^*\|_F}{d_{n0}} < 1$ then for all $(\vh_n,\vx_n) \in \setN_{d_{n0}}$. The following are helpful bounds that we have, $|\alpha_{n1}| < 2$, $|\alpha_{n2}| < 2$, and $|\alpha_{n1}\bar{\alpha}_{n2}-1| \leq \delta_n$.  For all $(\vh_n,\vx_n) \in \setN_{d_{n0}} \cap \setN_{\varepsilon}$ with $\varepsilon \leq 1/15$, there holds $\|\Delta \vh_n\|_2^2 \leq (4.6\delta_n^2 + 4\delta_0^2)d_{n0}$, $\|\Delta \vx_n\|_2^2 \leq (7.5\delta_n^2 + 2.88\delta_0^2)d_{n0}$, and $\|\Delta \vh_n\|_2^2 \|\Delta \vx_n\|_2^2 \leq \frac{1}{26} (\delta_n^2 + \delta_0^2) d_{n0}^2$. Moreover, if we assume $(\vh_n,\vx_n) \in \setN_{\mu} \cap \setN_{\nu}$, we have  $\sqrt{L}\|\mF_{M}\Delta\vh_n\|_\infty \leq 6 \mu \sqrt{d_{n0}}$, and $\sqrt{Q}\|\mC_n \Delta \vx_n \|_\infty \leq 6\nu \sqrt{d_{n0}}$. 
\end{lem}

Before stating the main lemma \ref{lem:main-lemma}, we define a neighborhood set as
\begin{align}\label{eq:setNF}
\setN_{F} : = \bigg\{(\vh,\vx) \ |  F(\vh,\vx) \leq \frac{1}{3N\kappa^2} \varepsilon^2d_0^2 + \|\ve\|_2^2\bigg\}
\end{align}
That neighborhood set is the sub-level of the objective function. We show that proposed gradient descent algorithm 1 makes the loss function to decrease; if the current iterate  $\vw^t \in \setN_{\varepsilon}\cap\setN_{F}$, then loss function goes on decreasing and the next iterate $\vw^{t+1}\in \setN_{\varepsilon}\cap\setN_{F}$\footnote{In literature, estimate at each iterate $t$, $\vw^t$ is represented by $\vz^t$ and vice versa}.  

\textbf{Proof of Lemma \ref{lem:local-regulaity-Del-norm-bounds}}\label{sec:local-regulaity-Del-norm-bounds}

We know that $\alpha_{n1} = \frac{\vh_n^* \vh_{n0}}{d_{n0}}$. By applying the Cauchy-Schwarz inequality and noting that $\vh_n \in \setN_{d_{n0}}$, we obtain $|\alpha_{n1}| \leq \frac{\|\vh_n\| \|\vh_{n0}\|}{d_{n0}} \leq 2$. Similarly, we can demonstrate that $|\alpha_{n2}| \leq 2$.
 Expand $\|\vh_n\vx_n^*-\vh_{n0}\vx_{n0}^*\|_F^2 = \delta_n^2d_{n0}^2$ using \eqref{eq:difference-expansion} to obtain
	\begin{align}
		\delta_n^2d_{n0}^2 = (\alpha_{n1}\bar{\alpha}_{n2}-1)^2d_{n0}^2 + |\bar{\alpha}_{n2}|^2 \|\tilde{\vh}_n\|_2^2 d_{n0} + |\alpha_{n1}|^2\|\tilde{\vx}_n\|_2^2 d_{n0} 
 +\|\tilde{\vh}_n\|_2^2\|\tilde{\vx}_n\|_2^2\notag,
	\end{align}
	which implies $|\alpha_{n1}\bar{\alpha}_{n2}-1| \leq \delta_n$. 
	 
	 The identities $\|\Delta \vh_n\|_2^2 \leq (4.6\delta^2_n + 4\delta^2_0)d_{n0}$, $\|\Delta \vx_n\|_2^2 \leq (7.5\delta^2_n + 2.88\delta^2_0)d_{n0}$, $\|\Delta \vh_n\|_2^2 \|\Delta \vx_n\|_2^2 \leq \frac{1}{26}(\delta^2_n + \delta^2_0)d_{n0}$, and $\sqrt{L}\|\mF_M\Delta\vh_n\|_\infty \leq 6 \mu \sqrt{d_{0}}$ are proved in Lemma 6.9 in \cite{li2018rapid}. We now prove that $\sqrt{Q} \|\mC \Delta \vx_n \|_\infty \leq 6\nu \sqrt{d_{0}}$.
	
	\textbf{Case 1}: $\|\vh_n\|_2 \geq \|\vx_n\|_2$, and $\alpha_n = (1-\delta_0)\alpha_{n1}$. Observe that in this case 
	\begin{align*}
	|\alpha_{n2}| &\leq \frac{\|\vx_n\|_2\|\vx_{n0}\|_2}{d_{n0}} \leq \frac{1}{\sqrt{d_{n0}}}\sqrt{\|\vh_n\|_2\|\vx_n\|_2} \\
	&\leq \frac{1}{\sqrt{d_{n0}}} \sqrt{\|\vh_n\vx_n^* - \vh_{n0}\vx_{n0}^*\|_F + \|\vh_{n0}\vx_{n0}^*\|_F} = \sqrt{1+\delta_n},
	\end{align*}
	where we used the fact that $\|\vh_n\vx_n^*-\vh_{n0}\vx_{n0}^*\|_F = \delta_n d_{n0}$, and $\|\vh_{n0}\|_2 = \|\vx_{n0}\|_2 = \sqrt{d_{n0}}$. Therefore, $\tfrac{1}{|(1-\delta_0)\alpha_{n1}|} = \tfrac{|\alpha_{n2}|}{|(1-\delta_0)\bar{\alpha}_{n2}\alpha_{n1}|} \leq \tfrac{\sqrt{1+\delta_n}}{|1-\delta_0||1-\delta_n|} \leq 2$, where the last inequality follows using our choice $\delta_n \leq \varepsilon \leq 1/15$, and $\delta_0 = \delta/10$. This gives us
	\begin{align*}
	\sqrt{Q} \|\mC \Delta\vx_n\|_\infty 
	&\leq \sqrt{Q}\|\mC \vx_n\|_\infty + \tfrac{1}{(1-\delta_0)|\alpha_{n1}|} \sqrt{Q} \|\mC_n\vx_{n0}\|_\infty\\
	& \leq 4\nu\sqrt{d_{0}}+2\nu \sqrt{d_{0}} \\
   & \leq 6\nu \sqrt{d_{0}}.
	\end{align*}
	
	\textbf{Case 2}: $\|\vh_n\|_2 < \|\vx_n\|_2$, and $\alpha_n = \frac{1}{(1-\delta_0)\bar{\alpha}_{n2}}$. Since $|\alpha_{n2}| \leq 2$, we have
	\begin{align*}
	\sqrt{Q}\|\mC\Delta\vx_n\|_\infty 
	&\leq \sqrt{Q}\|\mC \vx_n\|_\infty + (1-\delta_0)|\bar{\alpha}_{n2}| \sqrt{Q} \|\mC\vx_{n0}\|_\infty\\
	&\leq 4\nu\sqrt{d_{0}}+2(1-\delta_0)\nu \sqrt{d_{0}} \leq 6\nu \sqrt{d_{0}}.
	\end{align*}
	The proof is now complete. 

\begin{lem}[In \cite{li2018rapid}:Lemma 5.9]\label{lem:main-lemma}
	Suppose $\eta \leq 1/C_L$ is the step size , $\vw^t := (\vu^t,\mathbf{v}^t) \in \C^{N(M+K)}$, and Lipschitz constant, of $\nabla F (\vw^t)$ over $\setN_{d} \cap \setN_{\mu} \cap \setN_{\varepsilon}$, is $C_L$. If $\vw^t \in \setN_{\varepsilon} \cap \setN_{F}$, then $\vw^{t+1} \in \setN_{\varepsilon} \cap \setN_{F}$, and 
	\begin{align*}
	\tf(\vw^{t+1}) \leq \tf(\vw^t) - \eta \| \nabla F (\vw^t)\|_2^2. 
	\end{align*}
in which $\vw^{t+1} = \vw^{t} - \eta \nabla F (\vw^t).$
\end{lem} 
\begin{proof}
	It's proof is identical to the Lemma 5.9 proof in \cite{li2018rapid} that utilizes the smoothness condition as stated in Lemma \ref{lem:smoothness-CL}.
\end{proof}

\section{Key conditions} \label{sec:Key-conditions}
Four key conditions required to prove the theorems are stated below.

\subsection{Local regularity}\label{sec:local_regularity}

In Lemma \ref{lem:main-lemma}, a lower bound on $\|\gtf(\vw^t)\|^2_2$ is obtained by using the following lemma. 
\begin{lem}[Lemma 6.12 in \cite{li2018rapid}]\label{lem:local-regularity}
	Let $F(\vh,\vx)$ be as defined in \ref{eq:tF-def} and $\nabla F(\vh,\vx) : = (\nabla F_{\vh},\nabla F_{\vx}) \in \C^{N(M+K)}$. Then there exists a regularity constant $\omega = d_0/7000 >0$ such that 
	\begin{align*}
	\|\nabla F(\vh,\vx)\|_2^2 \geq \omega \left[ F(\vh,\vx) -c\right]_+
	\end{align*}
	for any $(\vh,\vx) \in \setN_{d_0} \cap \setN_{\mu} \cap \setN_{\nu} \cap \setN_{\varepsilon}$, where $c= \|\ve\|_2^2+2000 \|\setA^*(\ve)\|_{2\rightarrow 2}^2$, and $\rho \geq d^2+\|\ve\|_2^2$. 
\end{lem}
\begin{proof}
	If the following two conditions, that are stated in Lemma \ref{lem:local-regularity-F}, and \ref{lem:local-regularity-G}, 
	\begin{align*}
	&\Re{\<\nabla \tilde{F}_{\vh},\Delta\vh\>+\<\nabla \tilde{F}_{\vx},\Delta \vx\>} \geq \frac{\delta^2d_0^2}{8}-2\delta\sqrt{N} d_0 \|\setA^*(\ve)\|_{2\rightarrow 2},\notag \\ &\text{and} \notag\\
    &\Re{\<\nabla \tilde{G}_{\vh},\Delta\vh\>+\<\nabla \tilde{G}_{\vx},\Delta\vx\>} \geq \frac{\delta}{5}\sqrt{\rho G_0(\vh,\vx)};
	\end{align*}
are satisfied then the proof simplifies to the proof of Lemma 6.12 in \cite{li2018rapid}.
\end{proof}

\begin{lem}\label{lem:local-regularity-F}
	For any $(\vh,\vx) \in \setN_{d_0} \cap \setN_{\mu} \cap \setN_{\nu}\cap\setN_{\varepsilon}$ with $\varepsilon \leq \frac{1}{15}$:
	\[
	\Re{\< \Delta \tilde{F}_{\vh},\Delta\vh\>+\<\Delta \tilde{F}_{\vx},\Delta\vx\>} \geq \frac{\delta^2d_0^2}{8}-2\delta\sqrt{N} d_0 \|\setA^*(\ve)\|_{2 \rightarrow 2},
	\]
	with probability at least
	\begin{align}\label{eq:probability-regularity}
	1-2\exp\left(-c \epsilon^2 \delta_t^2\frac{Q}{\mu^2\nu^2 \kappa^4}\right)
	\end{align}
	provided 
	\begin{align}\label{eq:sample-complexity-regularity}
	Q \geq  \frac{c\kappa^4 N^2}{\epsilon^2 \delta_t^2}(\mu^2 \nu^2_{\max} K+\nu^2 M)\log^4(LN).
	\end{align}
\end{lem}

 \textbf{Proof of Lemma \ref{lem:local-regularity-F}}\label{sec:local-regularity-F}

	Note that $\Re{\<\nabla \tilde{F}_{\vh}, \Delta \vh\> + \<\nabla \tilde{F}_{\vx}, \Delta \vx\>} = \Re{\<\nabla \tilde{F}_{\vh}, \Delta \vh\> + \overline{\<\nabla \tilde{F}_{\vx}, \Delta \vx\>}}$. By utilizing the gradients previously derived in the paper in section III, we obtain
	\begin{align*} 
	&\<\nabla \tilde{F}_{\vh}, \Delta \vh\> = \< \setA^*(\setA(\mZ(\vh,\vx)-\mZ(\vh_0,\vx_0))-\ve),\mZ( \Delta\vh,\vx)\>,\\ & \text{and} \  \\
	& \overline{\<\nabla \tilde{F}_{\vx}, \Delta \vx\>} = \< \setA^*(\setA(\mZ(\vh,\vx)-\mZ(\vh_0,\vx_0))-\ve), \mZ(\vh,\Delta \vx)\>,
	\end{align*}
	and hence
	\begin{align}\label{eq:local-reg-F-expansion}
	&\<\nabla \tilde{F}_{\vh}, \Delta \vh\> + \overline{\< \nabla \tilde{F}_{\vx}, \Delta \vx\>} = - \<\setA^*(\ve), \mZ(\Delta\vh, \vx)+\mZ(\vh,\Delta\vx)\> \notag\\
	&\qquad + \<\setA(\mZ(\vh,\vx)-\mZ(\vh_0,\vx_0)),\setA(\mZ(\Delta\vh,\vx)+\mZ(\vh,\Delta\vx))\>.
	\end{align}
	By setting $\xi = \frac{1}{4}$ and $\delta \leq \varepsilon$ in Lemma \ref{lem:local-RIP}, and applying Lemma \ref{lem:local-Delh-Delx-RIP} below, we obtain the following findings:

	\begin{align*}
	\|\setA(\mZ(\vh,\vx)-\mZ(\vh_0,\vx_0))\|_2 & \geq \sqrt{\tfrac{3}{4}} \|\mZ(\vh,\vx)-\mZ(\vh_0,\vx_0)\|_F \\ &= \sqrt{\tfrac{3}{4}} \delta d_0, ~ \text{and} ~ \\
	\|\setA(\mZ(\Delta \vh,\vx)+ \mZ(\vh,\Delta\vx))\|_2 & \geq \sqrt{\tfrac{3}{4}} \|\mZ(\Delta \vh,\vx)+\mZ(\vh,\Delta\vx)\|_{F},
	\end{align*}
	each hold with minimum probability \eqref{eq:probability-regularity} given the complexity bound \eqref{eq:sample-complexity-regularity} is satisfied.  Utilizing triangle inequality, $\|\mZ(\Delta \vh,\vx)+ \mZ(\vh,\Delta\vx)\|_F \geq \|\mZ(\vh,\vx)-\mZ(\vh_0,\vx_0)\|_F - \|\mZ(\Delta\vh,\Delta\vx)\|_F$. Using Lemma \ref{lem:local-regulaity-Del-norm-bounds} it is easy to show $\|\mZ(\Delta \vh,\Delta \vx) \|_F \leq 0.2 \delta d_0$ when $\delta \leq \varepsilon \leq  1/15$. This implies that $\|\mZ(\Delta \vh,\vx)+ \mZ(\vh,\Delta\vx)\|_F  \geq \delta d_0 - 0.2 \delta d_0 \geq 0.8 \delta d_0.$
	Similarly, it is easy to show that $\|\mZ(\Delta \vh,\vx)+ \mZ(\vh,\Delta\vx)\|_F \leq 1.2 \delta d_0$. Additionally, we also have 
	\begin{align*}
	\<\setA^*(\ve), \mZ(\Delta\vh,\vx)&+\mZ(\vh,\Delta\vx)\> \\& \leq \|\setA^*(\ve)\|_{2 \rightarrow 2}\|\mZ(\Delta\vh,\vx)+ \mZ(\vh,\Delta\vx)\|_*\\
&\leq \sqrt{2N} \|\setA^*(\ve)\|_{2 \rightarrow 2}\|\mZ(\Delta\vh,\vx)+\mZ(\vh,\Delta\vx)\|_F\\
	& \leq  2\delta\sqrt{N} d_0 \|\setA^*(\ve)\|_{2\rightarrow 2}. 
	\end{align*}
	We acquire the intended bound by using the above bounds in \eqref{eq:local-reg-F-expansion}.

\begin{lem}\label{lem:local-regularity-G}
	Given any $(\vh,\vx) \in \setN_{\mu} \cap \setN_{\nu} \cap \setN_{d_0} \cap \setN_{\varepsilon}$ with $\varepsilon \leq 1/15$,  $0.9d_{n0} \leq d_n \leq 1.1d_{n0}$, and $0.9d_0 \leq d \leq 1.1d_0$, the next inequality applies uniformly $\Re{\<\nabla \tilde{G}_{\vh_n}, \Delta \vh_n \>+\<\nabla \tilde{G}_{\vx_n}, \Delta \vx_n \>}  \geq \frac{\delta}{5} \sqrt{\rho G_{n}(\vh_n,\vx_n)},$ where $\rho \geq d^2+ 2\|\ve\|_2^2$. Straight away, we can write $\Re{\<\nabla \tilde{G}_{\vh}, \Delta \vh \>+\<\nabla \tilde{G}_{\vx}, \Delta \vx \>}= \sum_{n=1}^N (\Re{\<\nabla \tilde{G}_{\vh_n}, \Delta \vh_n \>+\<\nabla \tilde{G}_{\vx_n}, \Delta \vx_n \>} )  \geq \frac{\delta}{5} \sqrt{\rho \tilde{G}(\vh,\vx)},$
\end{lem}

\textbf{Proof of Lemma \ref{lem:local-regularity-G}}\label{sec:local-regularity-G}

We know that $
	\Re{\<\nabla \tilde{G}_{\vh}, \Delta \vh \>+\<\nabla \tilde{G}_{\vx}, \Delta \vx \>} = \sum_{n=1}^N \Re{\<\nabla \tilde{G}_{\vh_n}, \Delta \vh_n \>+\<\nabla \tilde{G}_{\vx_n}, \Delta \vx_n \>} \geq \frac{\delta}{5} \sqrt{\rho G_0(\vh,\vx)}$, where each $G_n (\vh_n,\vx_n)$ is dependent on $(\vh_n,\vx_n)$ so it becomes the proof of Lemma 5.17 in \cite{li2016rapid}.

\textbf{Case 1}: $\|\vh_n\|_2 \geq \|\vx_n\|_2$, and $\alpha_n = (1-\delta_0) \alpha_{n1}.$ Given that $\delta_n \leq \varepsilon \leq 1/15$, we can derive the following identities, which can be easily verified (as shown in Lemma 5.18 of \cite{li2016rapid}): $\<\vh_n, \Delta \vh_n\> \geq \delta_0 \|\vh_n\|_2^2$, and $\|\vx_n\|_2^2 < 2d_n$. Additionally, we have 
\begin{align}\label{eq:inner-prod-bound-case1}
&\Re{\<\vf_{\ell}\vf_{\ell}^*\vh_n,\Delta \vh_n\>} \geq \frac{2d_n\mu^2}{L}   \ \text{when}  \ L \frac{|\vf_{\ell}^*\vh_n|^2}{8d_n\mu^2} > 1, \notag\\
& \Re{\<\vc_{q,n}\vc_{q,n}^*\vx_n,\Delta \vx_n\>} \geq \frac{d_n\nu^2}{Q} \ \text{when} \ Q \frac{|\vc_{q,n}^*\vx_n|^2}{8d_n\nu^2} > 1.
\end{align}
For instance, the last identity is easily proven as 
\begin{align*}
&\Re{\<\vc_{q,n}\vc_{q,n}^*\vx_n, \vx_n-\bar{\alpha}_n^{-1}\vx_{n0}\>}\\
&\geq |\vc_{q,n}^*\vx_n|^2 - \frac{1}{(1-\delta_0)|\alpha_{n1}|}|\vc_{q,n}^*\vx_n||\vc_{q,n}^*\vx_{n0}|\\
&= |\vc_{q,n}^*\vx_n|^2 - \frac{|\alpha_{n2}|}{(1-\delta_0)|\alpha_{n1}\bar{\alpha}_{n2}|}|\vc_{q,n}^*\vx_n||\vc_{q,n}^*\vx_{n0}|.
\end{align*}
By applying Lemma \ref{lem:local-regulaity-Del-norm-bounds}, we obtain the following: $|\alpha_{n2}| \leq 2$, $|\alpha_{n1}\bar{\alpha}_{n2} - 1| \leq \delta_n$, and the condition that $(\vh_n, \vx_n) \in \setN_\mu \cap \setN_\nu \cap \setN_{d_0}$, we also acquire 
\begin{align*}
&\text{Re}(\<\vc_{q,n}\vc_{q,n}^*\vx_n, \vx_n-\bar{\alpha}_n^{-1}\vx_{n0}\>)\\
&\geq |\vc_{q,n}^*\vx_n|^2  - \frac{2}{(1-\delta_n)(1-\delta_0)}|\vc_{q,n}^*\vx_n||\vc_{q,n}^*\vx_{n0}|\\
& \geq \frac{8d_n\nu^2}{Q} - \frac{2}{(1-\delta_n)(1-\delta_0)}\sqrt{\frac{8d_n\nu^2}{Q}}\cdot\sqrt{\frac{10d_n\nu^2 }{9Q}}\geq \frac{d_n\nu^2}{Q},
\end{align*}
where $|\vc_{q,n}^*\vx_{n0}| = \tfrac{\nu \sqrt{d_{n0}}}{\sqrt{Q}}$, and $0.9d_{n0} \leq d_n \leq 1.1 d_{n0}$ is used to derive the final inequality.

\textbf{Case 2}: $\|\vh_n\|_2 < \|\vx_n\|_2$, $\alpha_n = \tfrac{1}{(1-\delta_0)\bar{\alpha}_{n2}}$. For $\delta \leq \varepsilon \leq 1/15$, it can be demonstrated (see Lemma 5.17 in \cite{li2016rapid}) that
 $\<\vx_n,\Delta \vx_n\> \geq \delta_0 \|\vx_n\|_2^2, ~ \|\vh_n\|_2^2 < 2d_n$, and also 
\begin{align}\label{eq:inner-prod-bound-case2}
&\text{Re}(\<\vf_{\ell}\vf_{\ell}^*\vh_n,\Delta \vh_n\>) \geq \frac{d_n\mu^2}{L}  \ \text{when} \ L \frac{|\vf_{\ell}^*\vh_n|^2}{8d_n\mu^2} > 1,\notag\\
 &\Re{\<\vc_{q,n}\vc_{q,n}^*\vx_n,\Delta \vx_n\>}\geq \frac{2d_n\nu^2}{Q} \ \text{when} \ Q\frac{ |\vc_{q,n}^*\vx_n|^2}{8d_n\nu^2} > 1.
\end{align}
Expanding gradients make it clear that 
\begin{align}\label{eq:gGh-gGx-inner-prod}
&\Re{\<\nabla \tilde{G}_{\vh_n}, \Delta \vh_n \> + \<\nabla \tilde{G}_{\vx_n}, \Delta \vx_n \>} = \notag \\& \frac{\rho}{d_n} \Bigg( G_0^\prime\Big(\frac{\|\vh_n\|_2^2}{2d_n}\Big) \Re{\< \vh_n,\Delta \vh_n\>}+ G_0^\prime\Big(\frac{\|\vx_n\|_2^2}{2d_n}\Big) \Re{\< \vx_n,\Delta \vx_n\>}  \notag\\
&+ G_0^\prime \Big( \frac{L|\vf_{\ell}^*\vh_n|^2}{8d_n\mu^2} \Big) \frac{L}{4\mu^2} \Re{\<\vf_{\ell}\vf_{\ell}^*\vh_n, \Delta \vh_n\>}\notag\\
& + G_0^\prime \Big( \frac{Q|\vc_{q,n}^*\vx_n|^2}{8d_n\nu^2} \Big)\frac{Q}{4\nu^2} \Re{\<\vc_{q,n}\vc_{q,n}^*\vx_n, \Delta \vx_n\>} \Bigg). 
\end{align}
Now we figure out the subsequent inequality for the two aforementioned cases,
\begin{align*}
G_0^\prime \left(\frac{\|\vh_n\|_2^2}{2d_n}\right) \< \vh_n,\Delta \vh_n\> \geq \frac{\delta d_n}{5}G_0^\prime \left(\frac{\|\vh_n\|_2^2}{2d_n}\right).
\end{align*}
To demonstrate this, observe that it is obviously true when $\|\vh_n\|_2^2 < 2d_n$. In the opposite case, where $\|\vh_n\|_2^2 \geq 2d_n$, Case-2 cannot occur, and in Case-1 we have $\<\vh_n, \Delta\vh_n\> \geq \delta_0 \|\vh_n\|_2^2$. Thus, $\<\vh_n, \Delta\vh_n\> \geq \delta d_n / 5$ confirms that the inequality is valid. Likewise, we can also deduce that
\begin{align*}
 ~ G_0^\prime \left(\frac{\|\vx_n\|_2^2}{2d_n}\right) \< \vx_n,\Delta \vx_n\> \geq \frac{\delta d_n}{5}G_0^\prime \left(\frac{\|\vx_n\|_2^2}{2d_n}\right).
\end{align*}
Moreover, the following inequalities 
\begin{align}
 G_0^\prime \left( \frac{L|\vf_{\ell}^*\vh_n|^2}{8d_n\mu^2} \right) \frac{L}{4\mu^2} \text{Re}(\<\vf_{\ell}\vf_{\ell}^*\vh_n, \Delta \vh_n\> )\geq \frac{d_n}{4} G_0^\prime \left( \frac{L|\vf_{\ell}^*\vh_n|^2}{8d_n\mu^2} \right), \label{eq:G0-h-coherence-bound}\\
G_0^\prime \left( \frac{Q|\vc_{q,n}^*\vx_n|^2}{8d_n\nu^2} \right) \frac{Q}{4\nu^2} \text{Re}(\<\vc_{q,n}\vc_{q,n}^*\vx_n, \Delta \vx_n\> ) \geq \frac{d_n}{4} G_0^\prime \left( \frac{Q|\vc_{q,n}^*\vx_n|^2}{8d_n\nu^2} \right)\label{eq:G0-x-coherence-bound}
\end{align}
holds. To illustrate this point, observe that both inequalities hold trivially when $Q|\vc_{q,n}^*\vx_n|^2 > 8d_n\nu^2$ and $L |\vf_\ell^*\vh_n|^2 > 8d_n\mu^2$. Conversely, if these conditions are not met, using the bounds from \eqref{eq:inner-prod-bound-case1} and \eqref{eq:inner-prod-bound-case2}, we can deduce that \eqref{eq:G0-h-coherence-bound} and \eqref{eq:G0-x-coherence-bound} are satisfied in Cases 1 and 2. Substituting these outcomes into \eqref{eq:gGh-gGx-inner-prod} establishes the lemma.


\subsection{Local smoothness}\label{sec:local_smoothness}
In Lemma \ref{lem:main-lemma}, step size $\eta$ depends on the constant $C_L$. Following lemma qualify the $C_L$.
\begin{lem}\label{lem:smoothness-CL}
Considering any $\vw := (\vh, \vx)$ and $\vz := (\vu, \mathbf{v})$  in order to have $\vw$ and $\vw + \vz$ elements of $\setN_{\varepsilon} \cap \setN_{F}$, we have
\begin{align*}
&\|\nabla F(\vw + \vz) - \nabla F(\vw)\|_2 \leq C_L \|\vz\|_2 \quad \text{with} \\
& C_L \leq \sqrt{2}d_0 \left[ 10 \|\setA\|_{2 \rightarrow 2}^2 + \frac{\rho}{\min_n d_n^2} \left(5 + \frac{3L}{\mu^2} + \frac{3Q}{2\nu^2}\right) \right],
\end{align*}
where $\rho \geq d^2 + 2\|\ve\|_2^2$, and $\|\setA\|_{2 \rightarrow 2} \leq c_\alpha \sqrt{NK \log (L)}$ with minimum probability $1 - \setO(L^{-\alpha})$. Specifically, 
$Q = \setO(N^2(\mu^2 \nu^2_{\max} K + \nu^2 M))\log^4(LN)$, and $\|\ve\|^2 = \setO(\sigma^2 d_0^2)$. 
Consequently, $C_L$ can be reduced to
\begin{align}\label{eq:CL-def}
C_L = \setO\left(d_0 N^2(1 + \sigma^2) \big(\mu^2 \nu^2_{\max} K + \nu^2 M\big) \log^4(LN)\right)
\end{align}
by selecting $\rho \approx d^2 + 2 \|\ve\|_2^2$.

\end{lem}

\textbf{Proof of Lemma \ref{lem:smoothness-CL}}\label{sec:smoothness-CL}

Considering any $\vw= (\vh,\vx), \vz= (\vu,\mathbf{v})$, given the lemma below, and $\ \vw+ \vz= (\vh+\vu, \vx+\mathbf{v}) \in \setN_{F} \cap \setN_{\varepsilon}$, yields $\vw+\vz\in \setN_{d_0}\cap \setN_{\mu}\cap \setN_{\nu}$.  
	\begin{lem}
		Under local-RIP lemma \ref{lem:local-RIP} and noise robustness lemma \ref{lem:noise-stability}, there holds $\setN_{F} \subset \setN_{d_0} \cap \setN_{\mu}\cap \setN_{\nu}$. 
	\end{lem}
\noindent Proof of the above stated lemma follows exactly the same steps as the proof of Lemma 5.5 in \cite{li2018rapid}.  
	
	Upper bound of $\|\nabla \tilde{F}_{\vh}(\vw+\vz)-\nabla \tilde{F}_{\vh} (\vw)\|_2$ can be estimated by using gradient $\nabla \tilde{F}_{\vh}$ expansion. It is easy to see that 
		\begin{align*}
		\nabla  \tilde{F}_{\vh}(\vw+\vz)-  \nabla \tilde{F}_{\vh} (\vw) & = \setA^*\setA(\mZ (\vu, \vx)+ \mZ (\vh,\mathbf{v}) + \mZ (\vu,\mathbf{v}))\vx\\
		&  + \setA^*\setA(\mZ (\vh+\vu, \vx+\mathbf{v})-\mZ (\vh_0,\vx_0))\mathbf{v}-\setA^*(\ve)\mathbf{v}.
		\end{align*}
		
From Lemma \ref{lem:noise-stability}, we have $\|\setA^*(\ve)\|_{2 \rightarrow 2}\leq \varepsilon d_0$. Moreover, $\vw,\vw+\vz\in \setN_{d_0}$ implies $\|\mZ (\vu,\vx)+\mZ (\vh,\mathbf{v})+\mZ (\vu,\mathbf{v})\|_{F} \leq \|\vu\|_2\|\vx\|_2+\|\vh+\vu\|_2\|\mathbf{v}\|_2 \leq 2 \sqrt{d_0}(\|\vu\|_2+\|\mathbf{v}\|_2),$ where $\|\vh+\vu\|_2 \leq 2 \sqrt{d_0}$, and $\vw+\vz\in \setN_{\varepsilon}$ implies $\|\mZ (\vh+\vu,\vx+\mathbf{v})-\mZ (\vh_0,\vx_0)\|_F \leq \varepsilon d_0.$ Using these inequalities with $\|\vx\|_2 \leq 2 \sqrt{d_0}$ gives	
		
		\begin{align}\label{eq:grad-Fh-zw-z}
		\|\nabla \tilde{F}_{\vh}(\vw+\vz)-\nabla \tilde{F}_{\vh} (\vw)\|_2 & \leq 4d_0\|\setA\|^2_{2 \rightarrow 2} (\|\vu\|_2+\|\mathbf{v}\|_2) + 
		 \varepsilon d_0 \|\setA\|_{2 \rightarrow 2}^2 \|\mathbf{v}\|_2 + \varepsilon d_0 \|\mathbf{v}\|_2 \\ & \leq 5d_0\|\setA\|_{2 \rightarrow 2}^2 (\|\vu\|_2+\|\mathbf{v}\|_2).
		\end{align}
		Following inequality also hold due to symmetry between $\nabla \tilde{F}_{\vx}$ and $\nabla \tilde{F}_{\vh}$, 
		\begin{align}\label{eq:grad-Fx-zw-z}
		\|\nabla \tilde{F}_{\vx}(\vw+\vz)-\nabla \tilde{F}_{\vx} (\vw)\|_2\leq 5d_0 \|\setA\|_{2 \rightarrow 2}^2 (\|\vu\|_2+\|\mathbf{v}\|_2).
		\end{align}

To calculate the upper bound of $\| \nabla \tilde{G} (\vw+\vz) - \nabla \tilde{G}(\vw)\|_2$, we use the following equality
		\begin{multline}\label{eq:grad-main-Gh-zw-z}
		\| \nabla \tilde{G} (\vw+\vz) - \nabla \tilde{G}(\vw)\|_2 = \biggl\{ \sum_{n=1}^N \| \nabla \tilde{G}_{\vh_n} (\vw+\vz) - \nabla \tilde{G}_{\vh_n}(\vw)\|^2_2 + 
 \| \nabla \tilde{G}_{\vx_n} (\vw+\vz) - \nabla \tilde{G}_{\vx_n}(\vw)\|^2_2\biggl\}^{1/2}
		\end{multline}
where using \cite{ahmed2018ModBD} results for single modulated signal, gives
		\begin{align}\label{eq:grad-Gh-zw-z}
		\| \nabla \tilde{G}_{\vh_n} (\vw+\vz) - \nabla \tilde{G}_{\vh_n}(\vw)\|_2 \leq 5 \rho\frac{d_{n0}}{d_n^2}\|\vu_n\|_2 + \frac{3d_{n0}L\rho}{2d_n^2\mu^2}  \|\vu_n\|_2,
		\end{align}
and
		\begin{align}\label{eq:grad-Gx-zw-z}
		\| \nabla \tilde{G}_{\vx_n} (\vw+\vz) - \nabla \tilde{G}_{\vx_n}(\vw)\|_2 \leq 5 \rho\frac{d_{n0}}{d_n^2}\|\mathbf{v}_n\|_2 + \frac{3d_{n0}Q\rho}{2d_n^2\nu^2}  \|\mathbf{v}_n\|_2.
		\end{align}
So, \eqref{eq:grad-main-Gh-zw-z} becomes
		\begin{align}\label{eq:grad-main-simplified-Gx-zw-z}
		\| \nabla \tilde{G} (\vw+\vz) - \nabla \tilde{G} (\vw)\|_2 \leq \max \{ 5 \rho\frac{d_{n0}}{d_n^2} + \frac{3d_{n0}L\rho}{2d_n^2\mu^2}\}  \|\vu\|_2 + 
		  \max \{ 5 \rho\frac{d_{n0}}{d_n^2} + \frac{3d_{n0}Q\rho}{2d_n^2\nu^2}\}  \|\mathbf{v}\|_2.
		\end{align}				
	 Utilising $\|\vu\|_2+\|\mathbf{v}\|_2 \leq \sqrt{2}\|\vz\|_2$ as well as $\nabla F(\vw) = (\nabla \tilde{F}_{\vh} (\vw) + \nabla \tilde{G}_{\vh}(\vw), \nabla \tilde{F}_{\vx} (\vw) + \nabla \tilde{G}_{\vx}(\vw)),$ by plugging values of \eqref{eq:grad-Fh-zw-z}, \eqref{eq:grad-Fx-zw-z}, \eqref{eq:grad-Gh-zw-z}, and \eqref{eq:grad-Gx-zw-z}, gives	
	 \begin{align*}
	 \|\nabla F(\vw+\vz)-\nabla F(\vw)\|_2 \leq  \sqrt{2}\Big[10 d_0 \|\setA\|_{2 \rightarrow 2}^2+\frac{\rho}{\min_n d_n}\Big( 5+ \frac{3L}{2\mu^2} + \frac{3Q}{2\nu^2} \Big)\Big]\|\vz\|_2.
	 \end{align*}

\subsection{Local-RIP}\label{sec:local_RIP}
To proof Lemma \ref{lem:local-regularity-F}, the local-RIP that is described by the following two lemmas is used. 
\begin{lem}\label{lem:local-RIP}
	For all $(\vh_n,\vx_n) \in \setN_{d_0} \cap \setN_{\mu}\cap \setN_{\nu}$, and $\delta d_0 = \sqrt {\sum_{n=1}^N \|\vh_n\vx_n^*-\vh_{n0}\vx_{n0}^*\|_F^2}$, for a where $\xi = \tfrac{\varepsilon}{50\sqrt{N}\kappa} \in (0,1)$ the following local RIP holds:
	\begin{align}\label{eq:Local-RIP}
	 \Bigg|  \Big\|\sum_{n=1}^N  \setA_n(  \mZ_n(\vh_n,\vx_n)  -\mZ_n(\vh_{n0},\vx_{n0}))\Big\|_2^2 - \sum_{n=1}^N  \|\mZ_n(\vh_n,\vx_n)-\mZ(\vh_{n0},\vx_{n0})\|_{F}^2 \Bigg| \leq \xi \delta^2 d_0^2
	\end{align}
	 with probability at least $1-2\exp(-c\xi^2\delta^2 QN / \mu^2\nu^2\kappa^2)$
	  whenever 
	\begin{align}\label{eq:sample-complexity}
	Q \geq  \frac{c \kappa^2 N}{\xi^2 \delta^2 } \left(\mu^2 \nu^2_{\max} K + \nu^2 M \right) \log^4(LN).
	\end{align}

\end{lem}

\textbf{Proof of Lemma \ref{lem:local-RIP}}\label{sec:local-RIP}

In Fourier domain, we can write noiseless measurements in the form of linear map, $\setA_n$, defined in (4) as
\begin{align*}
\setA(\mZ) = \sum_{n=1}^N \vh_{n0}\circledast \mR_n\mC_n\vx_{n0} = \sum_{n=1}^N \text{circ}(\vh_{n0})\text{diag}(\mC_n\vx_{n0})\vr_n,
\end{align*}
where $\mR_n = \text{diag}(\vr_n)$. Let the matrices $\mH_{\vh},\mX_{\vx}$ are defined as
$$\mH_{\vh_n}:= [\text{circ}(\vh_n)], \ \text{and} \ \mX_{\vx_n} :=  [\text{diag}\left(\mC_n\vx_{n}\right)].$$ We write the error in terms of $l_2$-norm below
\begin{align}\label{eq:2nd-order-chaos-process}
\|\setA(\mZ-\mZ_0)\|_2^2 
  = \bigg\| \sum_{n=1}^N (\mH_{\vh_n}\mX_{\vx_n}-\mH_{\vh_{n0}}\mX_{\vx_{n0}})\vr_n \bigg\|_2^2.
\end{align}
The expected value of \eqref{eq:2nd-order-chaos-process} can be given as   
\begin{align}\label{eq:expected-value}
\E  \bigg\|\sum_{n=1}^N (\mH_{\vh_n}\mX_{\vx_n}- \mH_{\vh_{n0}}\mX_{\vx_{n0}})\vr_n \bigg\|_2^2  =   \sum_{n=1}^N \|  \mH_{\vh_n}\mX_{\vx_n}-\mH_{\vh_{n0}}\mX_{\vx_{n0}} \|_F^2,
\end{align}
Above equality is obtained as $\mH_{\vh_n}\mX_{\vx_n}-\mH_{\vh_{n0}}\mX_{\vx_{n0}}$ is deterministic and $\vr_n$ is a $Q$-length standard Rademacher vector. We can write $\text{circ}(\vh_n) = \mF^*\text{diag}(\hat{\vh}_n)\mF_Q \in \C^{L \times Q}$, where $\hat{\vh}_n = \sqrt{L}\mF_M\vh_n$ and $\mF$ is L$\times$L normalized DFT matrix. So, we can rewrite \eqref{eq:expected-value} as
\begin{align}\label{eq:distance-vectors-to-matrices}
\sum_{n=1}^N \|\mH_{\vh_n}\mX_{\vx_n}-\mH_{\vh_{n0}}\mX_{\vx_{n0}}\|_F^2 &= \sum_{n=1}^N  \|\mF^*\text{diag}(\hat{\vh}_n)\mF_Q \text{diag}(\mC_n\vx_{n})- \mF^*\text{diag}(\hat{\vh}_{n0})\mF_Q\text{diag}(\mC_n\vx_{n0})\|_F^2\notag\\
& = \sum_{n=1}^N \|\text{diag}(\hat{\vh}_n)\mF_Q \text{diag}(\mC_n\vx_{n})- \text{diag}(\hat{\vh}_{n0})\mF_Q\text{diag}(\mC_n\vx_{n0})\|_F^2\notag\\
& = \sum_{n=1}^N \| \mF_Q\odot( \hat{\vh}_n \vx_n^*- \hat{\vh}_{n0} \vx_{n0}^*)\|_F^2 = \sum_{n=1}^N \|\vh_n\vx_n^*-\vh_{n0}\vx_{n0}^*\|_F^2,
\end{align}
where $\mC_n^*\mC_n = \mI_K$,  $(\mF_M)^*\mF_M = \mI_M$, and $\sqrt{L}\mF_Q$ have normalized entries.

Now, we define the sets $\setH$, and $\setX$ as
\begin{align}\label{eq:setX-setH}
\setH:= \{\{\mH_{\vh_n}\}_{n=1}^N  \vert  \vh_n \in \setN_{d_0} \cap \setN_{\mu}\}, \notag \\
\ \setX: = \{ \{\mX_{\vx_n}\}_{n=1}^N  \vert    \vx_n \in \setN_{d_0} \cap \setN_\nu\}. 
\end{align}

By using \eqref{eq:2nd-order-chaos-process}, \eqref{eq:expected-value}, and \eqref{eq:distance-vectors-to-matrices}, we can write
\begin{align*}
 \sup_{\mH_{\vh} \in \setH}\sup_{\mX_{\vx} \in \setX}\bigg| &  \Big\|\sum_{n=1}^N (\mH_{\vh_n}\mX_{\vx_n}-\mH_{\vh_{n0}}\mX_{\vx_{n0}})\vr_n \Big\|_2^2 -  \E  \Big\| \sum_{n=1}^N (\mH_{\vh_n}\mX_{\vx_n}-\mH_{\vh_{n0}}\mX_{\vx_{n0}})\vr_n \Big\|_2^2 \bigg| \notag \\ &\leq \xi \sum_{n=1}^N\|\mH_{\vh_n}\mX_{\vx_n}-\mH_{\vh_{n0}}\mX_{\vx_{n0}}\|_F^2
\end{align*}
that is equivalent to local-RIP statement as stated in lemma \ref{lem:local-RIP} for a $\xi \in (0,1)$.   

Local-RIP proof follows a result from \cite{krahmer2014suprema}, as stated below, to see how much a second-order chaos process, $\sum_{n=1}^N \|(\mH_{\vh_n}\mX_{\vx_n}-\mH_{\vh_{n0}}\mX_{\vx_{n0}})\vr_n\|_2^2$, deviate from its mean.

\begin{thm}[In \cite{krahmer2014suprema}:Theorem 3.1]\label{thm:Mendelson}
	Suppose $\setG$ be a collection of matrices, and let a random vector $\vr$ whose components $r_j$ are independent random variables, have a mean of zero, a variance of $1$, and are $\alpha$-subgaussian. Following sets are defined as 
	\begin{align*}
	W &= (\gamma_2(\setG,\|\cdot\|_{2 \rightarrow 2})+d_{F}(\setG))\gamma_2(\setG,\|\cdot\|_{2 \rightarrow 2})  + d_{F}(\setG) d_{2\rightarrow 2}(\setG)\\
	V &= (\gamma_2(\setG, \|\cdot\|_{2\rightarrow 2})+d_F(\setG))d_{2 \rightarrow 2}(\setG), \ \text{and} \ U = d_{2 \rightarrow 2}^2 (\setG).
	\end{align*}
so, for $t \geq 0$, 
	\begin{align*}
	&\PP\big(\sup_{\mG \in \setG}\left| \|\mG\vr\|_2^2 - \E \|\mG\vr\|_2^2 \right| \geq c_1 W +t  \big) \leq \\
	& \qquad \qquad  2 \exp\Big( -c_2\min\Big\{ \tfrac{t^2}{V^2},\tfrac{t}{U}\Big\}\Big).
	\end{align*}
	where $c_1$, and $c_2$ constants are dependent on $\alpha$. 
\end{thm}

In above Theorem $\gamma_2$ is Talagrands functional \cite{talagrand2005generic} that defines how precisely, at different levels, a set $\setG$ can be approximated. For a set $\setG$, $d_{F}(\setG):= \sup_{\mG \in \setG} \|\mG\|_F$ defines the diameter w.r.t Frobenius norm, and  $d_{2\rightarrow 2}(\setG) := \sup_{\mG \in \setG} \|\mG\|_{2 \rightarrow 2}$ is the diameter w.r.t operator norm. We define the set of matrices as $\setG :=\{ [\mH_{\vh_n}\mX_{\vx_n}-\mH_{\vh_{n0}}\mX_{\vx_{n0}} ]^{\otimes N}| \mH_{\vh_n} \in \setH, \ \mX_{\vx_n} \in \setX \}$ where $\mH_{\vh_{n0}} \in \setH$, and $\mX_{\vx_{n0}} \in \setX$ are fixed. Using \eqref{eq:setX-setH}, we can see that $ \sup_{\mH_{\vh_n} \in \setH}\|  \mH_{\vh_n}^{\otimes N} \|_{2 \rightarrow 2} = \|[\text{circ}(\vh_n)]^{\otimes N}\|_\infty=\max_n \sqrt{L}\|\mF_M \vh_n\|_\infty \leq \max_n 4\mu \sqrt{d_{n0}}$, and $\|\mH_{\vh_0}^{\otimes N}\|_{2 \rightarrow 2}=\max_n \mu \sqrt{d_{n0}}$. For $\mX_{\vx_n} = \text{diag}( \mC_n\vx_{n})$, we have
\begin{align*}
\sup_{\mX_{\vx_n} \in \setX} \|\mX_{\vx_n}^{\otimes N}\|_{2 \rightarrow 2} = & \|\mC_n^{\otimes N}\vx\|_\infty \\ = & \max_{1\leq n \leq N}  \|\mC_n\vx_n\|_\infty \leq \tfrac{\max_n 4\nu \sqrt{d_{n0}}}{\sqrt{Q}}, \\
 \|\mX_{\vx_0}^{\otimes N}\|_{2 \rightarrow 2} =& \|\mC_n^{\otimes N}\vx_0\|_\infty \\ = & \max_{1\leq n \leq N} \|\mC_n\vx_{n0}\|_\infty = \tfrac{\max_n \nu \sqrt{d_{n0}}}{\sqrt{Q}}.
\end{align*}
The diameter of $\setG$ can be expressed using the Frobenius norm as 
\begin{align*}
 d_F(\setG) = \sup_{\mH_{\vh_n} \in \setH} \sup_{\mX_{\vx_n} \in \setX} \sqrt{ \sum_{n=1}^N \|\mH_{\vh_n}\mX_{\vx_n}-\mH_{\vh_{n0}}\mX_{\vx_{n0}}\|_F^2 } = \delta d_0.
\end{align*}

In the same way, diameter related to operator norm,
\begin{align*}
 d_{2 \rightarrow 2}(\setG) \leq  \sup_{\mH_{\vh_n}\in \setH} & \|\mH_{\vh_n}^{\otimes N}\|_{2 \rightarrow 2}  \cdot \sup_{\mX_{\vx_n}\in \setX} \|\mX_{\vx_n}^{\otimes N}\|_{2 \rightarrow 2}  + \|\mH_{\vh_{n0}}^{\otimes N}\|_{2 \rightarrow 2}\|\mX_{\vx_{n0}}^{\otimes N}\|_{2 \rightarrow 2} \\
\leq \frac{1}{\sqrt{Q}}\max_n & \left(4\mu\sqrt{d_{n0}}  \cdot 4\nu\sqrt{d_{n0}}+\mu\sqrt{d_{n0}}\cdot \nu \sqrt{d_{n0}}\right)\\ 
= \max_n & \frac{17\mu\nu d_{n0}}{\sqrt{Q}}. 
\end{align*}

Similar to \cite{ahmed2018ModBD}, we define $\epsilon$-cover of set $\setG$ as set $\setC$ defined below
\begin{align*}
\sup_{\mG\in \setG} \sup_{\mC \in \setC} \|\mC-\mG\| \leq \epsilon.
\end{align*}

Using Dudley integral, $\gamma_2$-functional can be written as \cite{dudley1967sizes}
\begin{align}\label{eq:Dudley-Integral}
\gamma_2(\setG,\|\cdot\|_{2 \rightarrow 2}) \leq c\int_0^{d_{2\rightarrow 2}(\setG)} \sqrt{ \log N(\setG,\|\cdot\|_{2 \rightarrow 2},\epsilon)d\epsilon },
\end{align}
where $N(\setG, \|\cdot\|, \epsilon)$ represents covering number of $\setG$ with smallest $\epsilon$-cover size, $d_{2 \rightarrow 2}(\setG)$ is the diameter w.r.t operator norm $\setG$, and $c$ is a known constant.

We define another set of matrices as $[\mH_{\tilde{\vh}_n}\mX_{\tilde{\vx}_n}-\mH_{\vh_{n0}}\mX_{\vx_{n0}}]^{\otimes N} \in \setC \subseteq \setG$, the distance between $\setC$, and $\setG$ can be given as 
\begin{align}\label{eq:distance}
& \|(\mH_{\tilde{\vh}_n}\mX_{\tilde{\vx}_n}-\mH_{\vh_{n0}}\mX_{\vx_{n0}})^{\otimes N} - (\mH_{\vh_n}\mX_{\vx_n}-\mH_{\vh_{n0}}\mX_{\vx_{n0}})^{\otimes N}\|_{2 \rightarrow 2} \notag  = \|(\mH_{\tilde{\vh}_n}\mX_{\tilde{\vx}_n}-\mH_{\vh_n}\mX_{\vx_n})^{\otimes N}\|_{2 \rightarrow 2}\notag \\
&\leq \|\mH_{\tilde{\vh}_n}^{\otimes N}\|_{2 \rightarrow 2}\|(\mX_{\tilde{\vx}_n}-\mX_{\vx_n})^{\otimes N}\|_{2 \rightarrow 2} \notag  + \|\mX_{\vx_n}^{\otimes N}\|_{2 \rightarrow 2}\|(\mH_{\tilde{\vh}_n}-\mH_{\vh_n})^{\otimes N}\|_{2 \rightarrow 2} \notag\\
&= \sqrt{L} \|\mF_{M}^{\otimes N} \tilde{\vh}\|_\infty \|\mC_n^{\otimes N}( \tilde{\vx}-\vx)\|_\infty   + \sqrt{L}  \|\mC_n^{\otimes N}\vx\|_\infty\|\mF_{M}^{\otimes N} (\tilde{\vh}-\vh)\|_\infty \notag \\
& \leq \max_n \frac{4\mu\sqrt{d_{n0}}}{\sqrt{Q}} \|\tilde{\vx}-\vx\|_c + \max_n \frac{4\nu \sqrt{d_{n0}} }{\sqrt{Q}} \|\tilde{\vh}-\vh\|_f.
\end{align}

To get $\|\mH_{\tilde{\vh}}\mX_{\tilde{\vx}}-\mH_{\vh}\mX_{\vx}\|_{2 \rightarrow 2} \leq \epsilon,$ we set the norms as
\begin{align}\label{eq:exotic-norms}
&\|\tilde{\vx}-\vx\|_{c}:= \sqrt{Q}\|\mC^{\otimes N} (\tilde{\vx}-\vx)\|_\infty \leq \frac{\epsilon}{2} \cdot \frac{\sqrt{Q}}{\max_n 4\mu\sqrt{d_{n0}}}, \notag\\
&\|\tilde{\vh}-\vh\|_{f} :=  \sqrt{L} \|\mF_{M}^{\otimes N} (\tilde{\vh} -\vh)\|_\infty \leq \frac{\epsilon}{2}\cdot \frac{\sqrt{Q}}{\max_n 4\nu \sqrt{d_{n0}}}.
\end{align}

If we plug the above choice of norms in \eqref{eq:distance}, then the point $[\mH_{\tilde{\vh}_n}\mX_{\tilde{\vx}_n}-\mH_{\vh_{n0}}\mX_{\vx_{n0}}]^{\otimes N}$ obeys $\| (\mH_{\tilde{\vh}_n}\mX_{\tilde{\vx}_n}-\mH_{\vh_{n0}}\mX_{\vx_{n0}})^{\otimes N} - (\mH_{\vh_n}\mX_{\vx_n}-\mH_{\vh_{n0}}\mX_{\vx_{n0}})^{\otimes N}\|_{2 \rightarrow 2}\leq \epsilon$. We conclude that an $\epsilon$-cover of $\setG$ is $\setC := \{[\mH_{\tilde{\vh}_n}\mX_{\tilde{\vx}_n}-\mH_{\vh_{n0}}\mX_{\vx_{n0}}]^{\otimes N} :  \mH_{\tilde{\vh}_n} \in \setC_{\setH}, \mX_{\tilde{\vx}_n} \in \setC_{\setX}\}$ in operator norm. 

Before calculating Dudley integral, we evaluate the covering number as follows
\begin{align*}
& N(\setG,\|\cdot\|_{2 \rightarrow 2}, \epsilon) \leq 
  N\bigg(\setX, \|\cdot\|_c,\frac{\epsilon \sqrt{Q}}{\max_n 8\mu\sqrt{d_{n0}}}\bigg)\cdot N\bigg(\setH, \|\cdot\|_f,\frac{\epsilon\sqrt{Q}}{\max_n 8\nu\sqrt{d_{n0}}}\bigg) \\ & \leq 
 N\bigg(2\sqrt{N d_{n0,\max}}B_2^{KN}, \|\cdot\|_c,\frac{\epsilon \sqrt{Q}}{\max_n 8\mu\sqrt{d_{n0}}}\bigg)\cdot  N\bigg(2\sqrt{N d_{n0,\max}}B_2^{MN}, \|\cdot\|_f,\frac{\epsilon\sqrt{Q}}{\max_n 8\nu\sqrt{d_{n0}}}\bigg) \\
& = N\bigg(B_2^{KN}, \|\cdot\|_c,\frac{\epsilon \sqrt{Q}}{\max_n 16\mu  \sqrt{N} d_{n0}}\bigg)\cdot  N\bigg(B_2^{MN}, \|\cdot\|_f,\frac{\epsilon\sqrt{Q}}{\max_n 16\nu  \sqrt{N} d_{n0}}\bigg)
\end{align*}
Now we calculate Dudley integral,
\begin{align*}
\int_0^{d_{2\rightarrow 2}(\setG)} & \sqrt{\log N(\setG,\|\cdot\|_{2 \rightarrow 2}, \epsilon)}d\epsilon \leq
\int_0^{\tfrac{\max_n 17\mu\nu d_{n0}}{\sqrt{Q}}} \Bigg(\sqrt{\log N\bigg(B_2^{KN}, \|\cdot\|_c,\frac{\epsilon \sqrt{Q}}{\max_n 16\mu \sqrt{N} d_{n0}}\bigg)} \\
&   + \sqrt{\log N\bigg(B_2^{MN}, \|\cdot\|_f,\frac{\epsilon\sqrt{Q}}{\max_n 16\nu \sqrt{N} d_{n0}}\bigg)}\Bigg) d\epsilon \\
&= \frac{\max_n 16\mu \sqrt{N} d_{n0}}{\sqrt{Q}} \int_{0}^{\frac{17\nu \sqrt{N}}{16}} \sqrt{\log N(B_2^{KN},\|\cdot\|_c,\epsilon)}d\epsilon + \\
& \qquad \qquad \frac{\max_n 16\nu \sqrt{N} d_{n0}}{\sqrt{Q}} \int_0^{\frac{17}{16}\mu \sqrt{N}} \sqrt{\log N(B_2^{MN},\|\cdot\|_f,\epsilon)}d\epsilon
\end{align*}
\begin{align*}
&\leq \frac{\max_n 16 \mu \sqrt{N} d_{n0}}{\sqrt{Q}} \sqrt{KN}\int_{0}^{ 2\nu \sqrt{N}} \sqrt{\log N(B_1^{KN},\|\cdot\|_c,\epsilon)}d\epsilon + \\
& \qquad  \frac{\max_n 16\nu \sqrt{N} d_{n0}}{\sqrt{Q}} \sqrt{MN}\int_0^{ 2\mu \sqrt{N}} \sqrt{\log N(B_1^{MN},\|\cdot\|_f,\epsilon)}d\epsilon\\
&  \lesssim \frac{\max_n \mu \nu_{\max}  \sqrt{N} d_{n0}}{\sqrt{Q}} \sqrt{KN \log^4 (QN)} +  \frac{\max_n \nu \sqrt{N} d_{n0}}{\sqrt{Q}} \sqrt{MN \log^4 (LN)},\\
&  \lesssim \frac{\kappa \mu \nu_{\max}   d_{0}}{\sqrt{Q}} \sqrt{KN \log^4 (QN)} +  \frac{\kappa \nu  d_{0}}{\sqrt{Q}} \sqrt{MN \log^4 (LN)},
\end{align*}
where we applied the knowledge of $B_2^{KN} \subseteq \sqrt{KN} B_1^{KN}$, and $B_2^{MN} \subseteq \sqrt{MN} B_1^{MN}$, in the third last inequality, and in the second last inequality, standard entropy is calculated (as given in Section 8.4 \cite{rauhut2010compressive}). In last inequality we used $\kappa = \frac{d_{\max}}{d_{\min}}$ and $d_{\min} \leq \frac{d_{0}}{\sqrt{N}}$. Using the result obtained above in \eqref{eq:Dudley-Integral} to get a bound on the $\gamma_2$ functional. Recall that $\delta = \tfrac{\sqrt{ \sum_{n=1}^N \|\vh_n\vx_n^*-\vh_{n0}\vx_{n0}\|^2_F}}{d_0} = \tfrac{\|\mX_{\vx}\mH_{\vh}-\mX_{\vx_0}\mH_{\vh_0}\|_F}{d_0}$. Now, we have all the required elements, i.e., diameters, and $\gamma_2$-functional, in Theorem \ref{thm:Mendelson}.   
\begin{align*}
 W \lesssim  d_0^2  \Bigg( \frac{\kappa^2}{Q} \left(\mu^2 \nu_{\max}^2 KN + \nu^2MN \right) \log^4(LN) + \frac{\delta \kappa}{\sqrt{Q}} \sqrt{\left(\mu^2 \nu_{\max}^2 KN + \nu^2MN \right) \log^4(LN)} +   \frac{ \mu\nu \delta  \kappa }{\sqrt{Q}}\Bigg).
\end{align*}
\begin{align*}
& U \lesssim \frac{\max_n \mu^2\nu^2 d_{n0}^2}{Q},  \\
&  V  \lesssim \frac{\max_n \mu\nu d_{n0}}{\sqrt{Q}} \Bigg( \frac{\kappa \mu \nu_{\max} d_{0}}{\sqrt{Q}} \sqrt{KN\log^4 (QN)}  + \frac{ \kappa \nu d_{0}}{\sqrt{Q}} \sqrt{MN \log^4 (LN)} + \delta d_0 \Bigg).
\end{align*}
Choosing $Q$ as in \eqref{eq:sample-complexity}, $L \geq Q$, $d^2_0 \geq Nd^2_{n0,min} $ and $t = \frac{1}{2} \xi d_0^2\delta^2$, the tail bound in Theorem \ref{thm:Mendelson} gives
\begin{align*}
\mathbb{P} \Bigg(\sup_{\mH_{\vh}\in \setH}\sup_{\mX_{\vx}\in \setX}\Bigg| \|(\mH_{\vh}\mX_{\vx}-\mH_{\vh_0}\mX_{\vx_0})\vr\|_2^2 -\|\mH_{\vh}\mX_{\vx}-\mH_{\vh_0}\mX_{\vx_0}\|_{F}^2 \Bigg| \geq  \xi\delta^2\d_0^2 \Bigg)\leq 2 \exp\left(-c \xi^2 \delta^2 \frac{QN}{\mu^2\nu^2 \kappa^2}\right).
\end{align*}

\begin{lem}\label{lem:local-Delh-Delx-RIP}
	For all $(\vh_n,\vx_n) \in \setN_{d_0} \cap \setN_{\mu}\cap \setN_{\nu}\cap \setN_{\varepsilon}$, and $\delta d_0 = \sqrt{ \sum_{n=1}^N \|\vh\vx^*-\vh_0\vx_0^*\|_F^2}$, where $\delta \leq \varepsilon \leq 1/15$, for a where $\xi = \tfrac{\varepsilon}{50\sqrt{N}\kappa} \in (0,1)$ the following local RIP holds:
	\begin{align*}
\bigg|  \Big\| \sum_{n=1}^N   \setA_n(\mZ_n(\Delta \vh_n, \vx_n) + \mZ_n(\vh_n, & \Delta \vx_n))\Big\|_2^2  
   - \sum_{n=1}^N \|\mZ_n(\Delta \vh_n,  \vx_n) + \mZ_n(\vh_n, \Delta \vx_n)\|_{F}^2\bigg| 
 \leq  \xi \delta^2 d_0^2
	\end{align*}
	with probability at least $1-2\exp(-c\xi^2\delta^2 QN / \mu^2\nu^2\kappa^2)$ whenever \eqref{eq:sample-complexity} holds.
\end{lem}
\textbf{Proof of Lemma \ref{lem:local-Delh-Delx-RIP}}\label{sec:local-Delh-Delx-RIP}

Already $\Delta \vh_n$, and $\Delta \vx_n$ are defined in \eqref{eq:Deltah-Deltax}, we also define $\Delta \mH_{\vh_n} = \mH_{\vh_n} - \alpha \mH_{\vh_{n0}}$, and $\Delta \mX_{\vx_n} = \mX_{\vx_n}-\bar{\alpha}^{-1}\mX_{\vx_{n0}}$. It is easy to show that $\sum_{n=1}^N \|\Delta \vh_n \vx_n^* + \vh_n \Delta \vx_n^*\|_F^2 = \sum_{n=1}^N \|\Delta \mH_{\vh_n} \mX_{\vx_n} + \mH_{\vh_n} \Delta \mX_{\vx_n}\|_F^2$ by using same steps as in \eqref{eq:distance-vectors-to-matrices}. Similar to the proof of Lemma \ref{lem:local-RIP}, the local-RIP in Lemma \ref{lem:local-Delh-Delx-RIP} holds with high probability and reduces to 
\begin{align*}
\sup_{\mH_{\vh_n} \in \setH}  \sup_{\mX_{\vx_n} \in \setX}\bigg| \Big \| \sum_{n=1}^N (\Delta \mH_{\vh_n} \mX_{\vx_n} + \mH_{\vh_n} \Delta \mX_{\vx_n})\vr_n \Big\|_2^2  & - \E  \Big\|\sum_{n=1}^N (\Delta \mH_{\vh_n} \mX_{\vx_n} + \mH_{\vh_n} \Delta \mX_{\vx_n})\vr_n \Big\|_2^2\bigg| \notag \\
& \leq \xi \sum_{n=1}^N \|\Delta \mH_{\vh_n} \mX_{\vx_n} + \mH_{\vh_n} \Delta \mX_{\vx_n}\|_{F}^2
\end{align*}
where a $0 < \xi < 1$. Just as $\setH$, and $\setX$ are defined in \eqref{eq:setX-setH}, define a set 
\begin{align}\label{eq:setS-2}
\setG =\{(\Delta \mH_{\vh_n} \mX_{\vx_n}  + \mH_{\vh_n} \Delta \mX_{\vx_n})^{\otimes N} \ |  \  \mH_{\vh_n}\in \setH, \mX_{\vx_n}\in \setX, (\vh_n, \vx_n) \in \setN_{\varepsilon}\},
\end{align}
 Let $(\Delta \mH_{\vh_n} \mX_{\vx_n} + \mH_{\vh_n} \Delta \mX_{\vx_n})^{\otimes N},$ and  $(\Delta \mH_{\tilde{\vh}_n} \mX_{\tilde{\vx}_n} + \mH_{\tilde{\vh}_n} \Delta \mX_{\tilde{\vx}_n})^{\otimes N}$ belongs to  $\setG$, and $(\Delta \mH_{\vh_n} \mX_{\vx_n} + \mH_{\vh_n} \Delta \mX_{\vx_n})^{\otimes N}-(\Delta \mH_{\tilde{\vh}_n} \mX_{\tilde{\vx}_n} + \mH_{\tilde{\vh}_n} \Delta \mX_{\tilde{\vx}_n})^{\otimes N} = (\Delta \mH_{\vh_n} - \Delta \mH_{\tilde{\vh}_n})^{\otimes N}\mX_{\tilde{\vx}_n}^{\otimes N} + \Delta \mH_{\vh_n}^{\otimes N} (\mX_{\vx_n}-\mX_{\tilde{\vx}_n})^{\otimes N} + (\mH_{\vh_n}-\mH_{\tilde{\vh}_n})^{\otimes N} \Delta \mX_{\vx_n}^{\otimes N} + \mH_{\tilde{\vh_n}}^{\otimes N}(\Delta\mX_{\vx_n}-\Delta\mX_{\tilde{\vx}_n})^{\otimes N}$, which gives 
\begin{align*}
\|(\Delta \mH_{\vh_n} & \mX_{\vx_n} + \mH_{\vh_n} \Delta \mX_{\vx_n})^{\otimes N} - (\Delta \mH_{\tilde{\vh}_n} \mX_{\tilde{\vh}_n} + \mH_{\tilde{\vh}_n} \Delta \mX_{\tilde{\vx}_n})^{\otimes N}\|_{2 \rightarrow 2} \\
&\leq \|(\mH_{\vh_n}-\mH_{\tilde{\vh}_n})^{\otimes N}\|_{2 \rightarrow 2} \|\mX_{\tilde{\vx}_n}^{\otimes N}\|_{2 \rightarrow 2}    + \|\Delta \mH_{\vh_n}^{\otimes N}\|_{2 \rightarrow 2}  \|(\mX_{\vx_n}-\mX_{\tilde{\vx}_n})^{\otimes N}\|_{2 \rightarrow 2}\\
& \qquad \qquad  +\|(\mH_{\vh_n}-\mH_{\tilde{\vh}_n})^{\otimes N}\|_{2 \rightarrow 2} \|\Delta \mX_{\vx_n}^{\otimes N}\|_{2 \rightarrow 2}  + \|\mH_{\tilde{\vh}_n}^{\otimes N}\|_{2 \rightarrow 2}  \|(\mX_{\vx_n}-\mX_{\tilde{\vx}_n})^{\otimes N}\|_{2 \rightarrow 2}  \\
&= \sqrt{L}\big(\|\mF_{M}^{\otimes N}(\vh-\tilde{\vh})\|_\infty \|\mC_n^{\otimes N}\tilde{\vx}\|_\infty   + \|\mF_{M}^{\otimes N}(\Delta \vh)\|_\infty  \|\mC_n^{\otimes N}(\vx-\tilde{\vx})\|_\infty  +\|\mF_{M}^{\otimes N}(\vh-\tilde{\vh})\|_\infty \|\mC_n^{\otimes N} (\Delta\vx)\|_\infty \notag \\ & \qquad \qquad  +\|\mF_{M}^{\otimes N}\tilde{\vh}\|_\infty  \|\mC_n^{\otimes N}(\vx-\tilde{\vx})\|_\infty\big),\\ &
\leq  \max_n  \Big(  \frac{4\nu\sqrt{d_{n0}}}{\sqrt{Q}} \sqrt{L}\|\mF_{M}^{\otimes N}(\vh-\tilde{\vh})\|_\infty +   6\mu\sqrt{d_{n0}}\|\mC_n^{\otimes N} (\vx-\tilde{\vx})\|_\infty     +  \frac{6\nu\sqrt{d_{n0}}}{\sqrt{Q}}  \sqrt{L}\|\mF_{M}^{\otimes N}(\vh-\tilde{\vh})\|_\infty  \\
& \qquad +  4\mu\sqrt{d_{n0}} \|\mC_n^{\otimes N}(\vx-\tilde{\vx})\|_\infty \Big) \\
&=  \max_n \frac{10\nu\sqrt{d_{n0}}}{\sqrt{Q}}  \|\vh-\tilde{\vh}\|_f+ \max_n \frac{10\mu\sqrt{d_{n0}}}{\sqrt{Q}} \|\vx-\tilde{\vx}\|_c,
\end{align*}

where to obtain second-last inequality we assumed $\varepsilon \leq 1/15$, and the elements $(\vh_n,\vx_n)$, $(\tilde{\vh}_n,\tilde{\vx}_n)$ of $\setG$  belongs to  $\setN_{d_0} \cap \setN_\mu \cap \setN_{\nu} \cap \setN_{\varepsilon}$, so using Lemma \ref{lem:local-regulaity-Del-norm-bounds} gives $\sqrt{L}\|\mF_{M}^{\otimes N}\Delta\vh\|_\infty \leq \max_n 6 \mu \sqrt{d_{n0}}$, and $\sqrt{Q}\|\mC_n^{\otimes N} \Delta \vx\|_\infty \leq \max_n 6 \nu \sqrt{d_{n0}}$. In the last equality, we used the $\|\cdot\|_c$, $\|\cdot\|_f$ norms as defined in \eqref{eq:exotic-norms}. The remaining proof follows precisely the same steps as Lemma \ref{lem:local-RIP}.


\subsection{Noise robustness}

A bound on $\|\setA^*(\ve)\|_{2 \rightarrow 2}$ that appear in Lemma \ref{lem:local-regularity} is given by noise robustness condition as stated in lemma below. 

\begin{lem}\label{lem:noise-stability}
	For fix $\alpha \geq 1$, $\|\setA\|_{2 \rightarrow 2} \leq c_\alpha\sqrt{ NK \log (L)}$ holds with minimum probability $1-\setO(L^{-\alpha})$. Additionally, suppose additive noise $\ve \sim \text{Normal}(\boldsymbol{0},\frac{\sigma^2d_0^2}{2L} \mI_{L}) + \iota \text{Normal}(\boldsymbol{0},\frac{\sigma^2d_0^2}{2L} \mI_{L})\in \C^{L}$, with minimum probability $1-\setO(L^{-\alpha})$
	\begin{align}\label{eq:A(e)-bound}
	\|\setA^*(\ve) \|_{2 \rightarrow 2} \leq \tfrac{2\varepsilon}{50N \kappa} d_0
	\end{align}
	holds whenever $L \geq  \frac{\sigma^2 N^2 \kappa^2}{\varepsilon^2} c_\alpha^\prime \max (M,K\log(L))\log(L)$ where $c_\alpha$, and $c^\prime_\alpha$ are absolute constants that depends on $\alpha$. 
\end{lem}

\textbf{Proof of Lemma \ref{lem:noise-stability}}\label{sec:noise-stability}
Proof of this lemma consists of two parts: first we find a bound on operator norm of the linear map $\|\setA\|_{2 \rightarrow 2}$, and after that we prove the claim stated in the second part of lemma.

So we start by proving the first claim of the lemma. Let $\vf_\ell$ are the columns of $\mF_M^*$, and $\hat{\vc}_{\ell,n}^*$ are rows of $\sqrt{L}(\mF_Q\mR_n \mC_n)$ then $\< \sum_{n=1}^N \vf_\ell\hat{\vc}_{\ell,n}^*, \sum_{n=1}^N \vf_{\ell^\prime}\hat{\vc}_{\ell^\prime,n^\prime}^*\> = 0$ once $\ell \neq \ell^\prime$ so $\|\setA\|_{2 \rightarrow 2} = \max_{\ell} \| \sum_{n=1}^N \vf_\ell\hat{\vc}_{\ell,n}^*\|_F $. We see that $\| \sum_{n=1}^N \vf_\ell\hat{\vc}_{\ell,n}^*\|_F =  \| \sum_{n=1}^N \hat{\vc}_{\ell,n}\|_2$ all that's needed is an upper constraint on $\max_{\ell} \|\hat{\vc}_{\ell,n}\|_2$.

As already defined, Rademacher random variables of length $Q$ represented as $\vr_n$ so we can write a diagonal matrix $\mR_n = \text{diag}(\vr_n)$. Let $\vc_{q,n}^*$ are the rows of $\mC_n$ and $q$th entry of $\vr_n$ is $r_n[q]$ so we can write
	\begin{align*}
	\hat{\vc}_{\ell,n} = \sqrt{L}\sum_{q=1}^Q f_{\ell}[q]r_n[q]\vc_{q,n}.
	\end{align*}
By using Proposition \ref{prop:conc_ineq}, we obtain an upper bound on $\|\hat{\vc}_{\ell,n}\|_2$ assuming $\sqrt{L}\{f_{\ell}[q]r_n[q]\vc_{q,n}\}_{q=1}^Q$ is sequence $\{\mG_k\}$. We know that $\mC_n$ is orthonormal matrix of dimension $Q \times K$ so $\sum_{q=1}^Q \|\vc_{q,n}\|_2^2 = K$ and $\sum_{q=1}^Q \vc_{q,n}\vc_{q,n}^* = \mI$. By using \eqref{eq:variance-conc-ineq}, calculated variance is $\sigma_Z^2 \leq K+1$. Using \eqref{eq:conc-ineq-bound} we can write that with minimum probability $1-\setO(L^{-\alpha})$
	\begin{align}\label{eq:hatbln-2norm-bound}
	\max_{\ell,n}\|\hat{\vc}_{\ell,n}\|_2  \leq \sqrt{\alpha K\log(L)}
	\end{align}
where $ t^2 = \alpha K \log(L)$ and $\alpha \geq 1$. 
 This leads to 
	\begin{align}\label{eq:hatb_ln-2norm-bound}
	\|\setA\|  \leq \sqrt{\alpha NK \log(L)}
	\end{align}
	with probability at least $1-\setO(L^{-\alpha})$.
	
	Now we prove the second claim, by using the normal Gaussian random variables $p[\ell] \sim \text{Normal}(0,\tfrac{1}{2}) + \iota \text{Normal}(0,\tfrac{1}{2})$ to rewrite the Gaussian random variables $\hat{e}[\ell]$. So, 
	\begin{align*}
	\setA_n^*(\ve) =  \sum_{\ell=1}^L   \hat{e}[\ell]\hat{\vc}_{\ell,n}\vf^*_\ell =   \frac{\sigma d_0}{\sqrt{L}} \sum_{\ell=1}^L g[\ell] \hat{\vc}_{\ell,n}\vf^*_\ell,
	\end{align*}

To control $\|\setA_n^*(\ve)\|_{2 \rightarrow 2}$ we use matrix concentration inequality in Proposition \ref{prop:conc_ineq}. Matrices  $\{\hat{\vc}_{\ell,n}\vf^*_\ell\}_{\ell,n}$ are represented as $\{\mZ_k\}$  in Proposition \ref{prop:conc_ineq}. Now, variance in \eqref{eq:variance-conc-ineq} is 
	\begin{align*}
	\sigma^2_Z = \frac{\sigma^2 d_0^2}{L} \max_n \Bigg\{\left\|\sum_{\ell=1}^L  \|\vf_\ell\|_2^2 \hat{\vc}_{\ell,n} \hat{\vc}_{\ell,n}^*\right\|_{2 \rightarrow 2},  \left\|\sum_{\ell=1}^L  \|\hat{\vc}_{\ell,n}\|_2^2 \vf_{\ell} \vf_{\ell}^*\right\|_{2 \rightarrow 2} \Bigg\}.
	\end{align*}
 
	As we assumed that $\vf_\ell$ are the columns of $\mF_M^*$, rows of $\sqrt{L}(\mF_Q\mR_n \mC_n)$ are denoted as $\hat{\vc}_{\ell,n}^*$, so 
	\begin{align*}
	&\|\vf_\ell\|_2^2 = \frac{M}{L}, \ \sum_{\ell=1}^L \vf_\ell\vf_\ell^* = \mI,\\
	&\sum_{\ell=1}^L \hat{\vc}_{\ell,n} \hat{\vc}_{\ell,n}^* = L (\mC_n^*\mR_n^*(\mF_Q)^*\mF_Q \mR_n\mC_n) = L \mI_{K\times K}. 
	\end{align*}
	Using the above equations with with \eqref{eq:hatbln-2norm-bound}, upper bound on the variance is 
	\begin{align*}
	\sigma^2_Z \leq \frac{\sigma^2d_0^2}{L} \max \left( M, \alpha K \log(L) \right). 
	\end{align*}
	Using  inequality \eqref{eq:conc-ineq-bound} in Proposition \ref{prop:conc_ineq} and $t = \tfrac{2\varepsilon d_0}{50N\kappa}$, gives 
	\[
	L \geq \frac{\sigma^2 N^2 \kappa^2}{\varepsilon^2} c^\prime_{\alpha} \max(M,\alpha K \log(L))\log(L)
	\]
	 which brings the proof to a completion.

\begin{prop}
	[In \cite{tropp12us}: Corollary 4.2 ]\label{prop:conc_ineq} Let $\{\mG_i\}$ be a $d_1 \times d_2$ dimensional finite series of fixed matrices, and consider a finite series $\{p_i\}$ of independent random variables, which can be either Gaussian or Rademacher. We define the variance as

	\begin{align}\label{eq:variance-conc-ineq}
	\sigma_G^2 := \max \left\{ \left\|\sum_i \mG_i\mG_i^*\right\|_{2 \rightarrow 2}, \left\|\sum_i \mG_i^*\mG_i\right\|_{2 \rightarrow 2} \right\}. 
	\end{align}
	Subsequently, for all $t \geq 0$
	\begin{align}\label{eq:conc-ineq-bound}
	\mathbb{P} \left(\left\| \sum_i p_i \mG_i\right\|_{2 \rightarrow 2} \geq \ \ t \right) \leq (d_1+d_2)\mathrm{e}^{-t^2/2\sigma_G^2}.
	\end{align}
\end{prop}

We now define the bounds on $F(\vh,\vx)$, given in (8), using local-RIP and noise robustness conditions. Using triangle inequality on $F(\vh,\vx)$ gives
\begin{align*}
 - 2 \|\setA^*(\ve)\|_{2 \rightarrow 2}  \|\mZ(\vh,\vx)  -\mZ(\vh_0,\vx_0)\|_* & \leq  
 F(\vh,\vx)  -\|\ve\|_2^2  - \|\setA(\mZ(\vh,\vx) - \mZ(\vh_0,\vx_0))\|_2^2 \notag\\ & \leq  2 \|\setA^*(\ve)\|_{2 \rightarrow 2}\|\mZ(\vh,\vx)-\mZ(\vh_0,\vx_0)\|_*.
\end{align*}
where $\|\mZ(\vh,\vx)-\mZ(\vh_0,\vx_0)\|_F : = \delta d_0$. We get $\|\mZ(\vh,\vx)-\mZ(\vh_0,\vx_0)\|_* \leq \sqrt{2N} \|\mZ(\vh,\vx)-\mZ(\vh_0,\vx_0)\|_F = \sqrt{2N}\delta d_0$ by using that matrix has rank $2N$ and $\|\vh_{n0}\|_2 = \|\vx_{n0}\|_2 = \sqrt{d_{n0}}$ from Lemma \ref{lem:local-regulaity-Del-norm-bounds}.
Using Lemma \ref{lem:local-RIP} and \ref{lem:noise-stability},  with $\xi = \frac{1}{4}$, gives
\begin{align}\label{eq:F(h,x)-lower-upper-bound}
\|\ve\|_2^2 +\frac{3}{4} \delta^2 d_0^2 - \frac{\varepsilon\delta d_0^2}{5\sqrt{N}\kappa} \leq F(\vh,\vx) \leq \|\ve\|_2^2 +\frac{5}{4} \delta^2 d_0^2 + \frac{\varepsilon\delta d_0^2}{5\sqrt{N}\kappa}.
\end{align}


\section{Proof of Theorem 1}\label{sec:convergence}
We have stated all the conditions required to prove theorems. Now we prove Theorem 1 stated in section IV of the paper. 
\begin{proof}
At the $t$-th iteration of gradient descent algorithm, $\{\vw_n^t\}_{n=1}^N = (\vu_n^t,\mathbf{v}_n^t)$ where $\vw_n = \vu_n \mathbf{v}_n^*$ and $\delta_n(\vw_n^t) = \|\vw_n^t-\vh_{n0}\vx_{n0}^*\|_F/d_{n0}$.
Initial guess $\{\vw_n^0 \}_{n=1}^N :=(\vu_n^0,\mathbf{v}_n^0) \in \tfrac{1}{\sqrt{3}} \setN_{d_0} \cap \tfrac{1}{\sqrt{3}} \setN_{\mu} \cap \tfrac{1}{\sqrt{3}} \setN_{\nu} \cap \setN_{\frac{2\varepsilon}{5\sqrt{N}\kappa}}$, so $G(\vu^0,\mathbf{v}^0) = 0$. To check this, let $(\vu_n^0,\mathbf{v}_n^0) \in \frac{1}{\sqrt{3}}\setN_{\nu}$ then 
\begin{align*}
\frac{Q|\vc_{q,n}^*\vx_n|^2}{8d_n\nu^2} \leq \frac{Q}{8d_n\nu^2} \cdot \frac{16d_{n0}\nu^2}{3Q} = \frac{2d_{n0}}{3 d_n} < 1,
\end{align*}
which shows the last term of $G_n(\vu_n^0,\mathbf{v}_n^0)$ in (10) becomes zero and similarly, for this initialization all the other terms of $G_n(\vu_n^0,\mathbf{v}_n^0)$ becomes zero. As a result, $G(\vu^0,\mathbf{v}^0)$ in (9) becomes zero. The remaining portion of the proof is the same as that of Theorem 3.3 in \cite{li2018rapid}. For step size $\eta \leq 1/C_L$ in algorithm 1, using Lemma \ref{lem:main-lemma}, and \ref{lem:local-regularity} we get the following inequality
\begin{align*}
\|\mZ(\vu^t,\mathbf{v}^t)-\mZ(\vh_0,\vx_0)\|_F \leq  \tfrac{\varepsilon d_0}{\sqrt{2N\kappa}} (1-\eta\omega)^{t/2} \varepsilon d_0  + 60 \sqrt{N}\|\setA^*(\ve)\|_{2 \rightarrow 2}.
\end{align*}

\end{proof}
\section{Proof of Theorem 2}\label{sec:Theorem2-Proof}
To obtain a good initial guess: $\{\vu_{n0},\mathbf{v}_{n0}\}_{n=1}^N \in \frac{1}{\sqrt{3}}\setN_{d_0} \cap \frac{1}{\sqrt{3}}\setN_{\mu}  \cap \frac{1}{\sqrt{3}}\setN_{\nu}\cap \setN_{\frac{2\varepsilon}{5\sqrt{N}\kappa}}$, from observation $\vy$, and the model $\setA$, we now proof the Theorem 2 stated in section IV of the paper. 

\begin{proof} 
	For a choice $(\vh,\vx) = (\boldsymbol{0},\boldsymbol{0})$, $\delta d_0 = \|\mZ(\vh,\vx)-\mZ(\vh_0,\vx_0)\|_F = \|\mZ(\vh_0,\vx_0)\|_F = d_0$ gives $\delta = 1$. For $\delta = 1$, Lemma \ref{lem:local-RIP} gives $(1-\xi) \|\mZ(\vh_0,\vx_0)\|_F^2 \leq \|\setA(\mZ(\vh_0,\vx_0)\|_2^2 \leq (1+\xi) \|\mZ(\vh_0,\vx_0)\|_F^2$ with minimum probability $ 1-2\exp\left(-c\xi^2QN/\mu^2\nu^2\kappa^2\right),$ and sample complexity, for $\xi$-RIP to hold,  
	\begin{align}\label{eq:sample-complexity-initialization}
	Q \geq \frac{c\kappa^2 N}{\xi^2}\left( \mu^2 \nu^2_{\max} K + \nu^2 M\right)\log^4(LN).
	\end{align}
    We can restate the $\xi$-RIP condition as 
\begin{align*}
|\< (\setA_n^*\setA_n-\setI)(\vh_{n0}\vx_{n0}^*), \vh_{n0}\vx_{n0}^* \>| 
&\leq \xi \|\vh_{n0}\vx_{n0}^*\|_F^2,
\end{align*}
and hence $\|\setA_n^*\setA_n(\vh_{n0}\vx_{n0}^*) - \vh_{n0}\vx_{n0}^*\|_{2 \rightarrow 2} \leq \xi d_{n0}.$  By applying triangle inequality gives,
\begin{align}\label{eq:A*(y)-h0x0-operator-norm}
\|\setA_n^*(\hat{\vy})-\vh_{n0}\vx_{n0}^*\|_{2\rightarrow 2} \leq   \|\setA_n^*\setA_n(\vh_{n0}\vx_{n0}^*) -\vh_{n0}\vx_{n0}^*\|_{2 \rightarrow 2} + \|\setA_n^*(\vw_n)\|_{2 \rightarrow 2}
\end{align}
where $\vw_n = \sum_{m \neq n} \setA_m(\vh_{m0}\vx_{m0}^*) +\ve$. Choosing $L \geq \frac{(\mu_h^2  + \sigma^2) N^2 \kappa^4}{\varepsilon^2} c^\prime_{\alpha} \max(M,\alpha K \log(L))\log(L)$, we can write \eqref{eq:A*(y)-h0x0-operator-norm} as
\begin{align}\label{eq:A*(y)-h0x0-operator-norm-simplified}
\|\setA_n^*(\hat{\vy})-\vh_{n0}\vx_{n0}^*\|_{2\rightarrow 2}
 \leq \xi d_{n0} +  \frac{2\varepsilon d_{n0}}{50\sqrt{N}\kappa} \leq  \frac{3\varepsilon d_{n0}}{50\sqrt{N}\kappa} := \gamma d_{n0},
\end{align}
by using $t=\frac{2\epsilon d_{n0}}{50\sqrt{N}\kappa}$ from Lemma \ref{lem:noise-stability},  and choosing $\xi = \tfrac{\varepsilon}{50\sqrt{N}\kappa}$ we get the last inequality. Recall that the leading singular value of $\setA_n^*(\hat{\vy})$ is denoted as $d_n$, $\hat{\vx}_{n0}$, and $\hat{\vh}_{n0}$ shows the corresponding right, and left singular vectors, respectively. Thus, we can write $|d_n- d_{n0}| \leq \tfrac{3\varepsilon}{50\sqrt{N}\kappa} d_{n0}$ and we conclude that $0.9 d_{n0} \leq d_n \leq 1.1d_{n0}$ with $\varepsilon \leq \tfrac{1}{15}$. 

In Algorithm 2, 
$\sqrt{d_n} \vx_{n0}$ is projected onto the convex set $\setZ = \{\vz | \sqrt{Q} \|\mC_n\vz\|_\infty \leq 2 \sqrt{d_n} \nu \}$ to get initialization $\mathbf{v}_{n0}$ of $\vx_{n0}$.
Now $\mathbf{v}_{n0} \in \setZ$ implies that $\sqrt{Q}\|\mC_n \mathbf{v}_{n0}\|_\infty \leq 2 \sqrt{d_n} \nu \leq \frac{4\nu}{\sqrt{3}}$, and hence $\{\mathbf{v}_{n0}\}_{n=1}^N \in \frac{1}{\sqrt{3}} \setN_{\nu}$ for all $1\leq n \leq N$. We can write it as,
\begin{align}\label{eq:projection-lemma-result}
\|\sqrt{d_n}\hat{\vx}_{n0}- \vp\|_2^2 & = 
\|\sqrt{d_n}\hat{\vx}_{n0} - \mathbf{v}_{n0}\|_2^2 + 2\Re{\<\hat{\vx}_{n0}-\mathbf{v}_{n0},\mathbf{v}_{n0}-\vp\>}+\|\mathbf{v}_{n0}-\vp\|_2^2\notag\\
&\geq \|\sqrt{d_n}\hat{\vx}_{n0} - \mathbf{v}_{n0}\|_2^2+\|\mathbf{v}_{n0}-\vp\|_2^2
\end{align}
for all $\vp \in \setZ$, where Lemma \ref{lem:convex-set-projection} is used on the inner product. Above inequality gives $\|\mathbf{v}_{n0}\|_2 \leq \sqrt{d_n} \leq \frac{2}{\sqrt{3}}$ by choosing $\vp = \boldsymbol{0} \in \setZ$, and hence $\{\mathbf{v}_{n0}\}_{n=1}^N \in \frac{1}{\sqrt{3}}\setN_{d_0}$. We have thus shown that $\{\mathbf{v}_{n0}\}_{n=1}^N \in \frac{1}{\sqrt{3}}\setN_{d_0} \cap \frac{1}{\sqrt{3}}\setN_{\nu}$. Similarly, we can show that $\{\vu_{n0}\}_{n=1}^N \in \frac{1}{\sqrt{3}}\setN_{d_0} \cap \frac{1}{\sqrt{3}}\setN_{\mu}$. 

Now, we show that $\{\vu_{n0},\mathbf{v}_{n0}\}_{n=1}^N \in \setN_{\frac{2\varepsilon}{5\sqrt{N}\kappa}}$ by starting with $\|\setA_n^*(\hat{\vy})-\vh_{n0}\vx_{n0}^*\|_{2 \rightarrow 2} \leq d_{n0} \gamma$, let $i$th leading singular value of $\setA_n^*(\hat{\vy})$ is represented as $\sigma_i(\setA_n^*(\hat{\vy}))$, this indicates that $\sigma_i(\setA_n^*(\hat{\vy})) \leq \gamma$ for $i \geq 2$. Applying \eqref{eq:A*(y)-h0x0-operator-norm} and utilizing the triangle inequality, we obtain

\begin{align}\label{eq:distnace-initialization}
\|d_n \hat{\vh}_{n0}\hat{\vx}_{n0}^* - \vh_{n0} \vx_{n0}^*\|_{2 \rightarrow 2} \leq \|\setA_n^*(\hat{\vy}) - d_n \hat{\vh}_{n0}\hat{\vx}_{n0}^*\|_{2 \rightarrow 2} +  \|\setA_n^*(\hat{\vy}) - \vh_{n0}\vx_{n0}^*\|_{2 \rightarrow 2} \leq 2 d_{n0}\gamma,
\end{align}
We also have 
\begin{align*}
 \left\|\hat{\vx}_{n0}^*(\mI-\frac{\vx_{n0}\vx_{n0}^* }{d_{n0}})\right\|_2 = & \left\|(\hat{\vx}_{n0}\hat{\vh}_{n0}^* \hat{\vh}_{n0} \hat{\vx}_{n0}^*)(\mI-\frac{\vx_{n0}\vx_{n0}^*}{d_{n0}}) \right\|_F \\
 = & \left\| \frac{\hat{\vx}_{n0} \hat{\vh}_{n0}^*}{d_{n0}}(\setA_n^*(\vy) - d_n \hat{\vh}_{n0}\hat{\vx}_{n0}^* + \hat{\vh}_{n0}\hat{\vx}_{n0}^* - \vh_{n0}\vx_{n0}^*) (\mI-\frac{\vx_{n0}\vx_{n0}^*}{d_{n0}})\right\|_F\\
 \leq & \left\|\hat{\vx}_{n0}\hat{\vh}_{n0}^* \frac{1}{d_{n0}}(\setA_n^*(\vy)-\vh_{n0}\vx_{n0}^*)(\mI-\frac{\vx_{n0}\vx_{n0}^*}{d_{n0}})\right\|_F + |\frac{d_n}{d_{n0}}-1| \\  \leq & 2\gamma,
\end{align*}
where we used $\vh_{n0}\vx_{n0}^*(\mI-\frac{\vx_{n0}\vx_{n0}^*}{d_{n0}}) = \mathbf{0}$ to get the second equality, and $\hat{\vx}_{n0}\hat{\vh}_{n0}^*(\setA_n^*(\vy)-d_n\hat{\vh}_{n0}\hat{\vx}_{n0}^*) = \mathbf{0}$. Let $\beta_0 = \frac{\sqrt{d_n}}{d_{n0}}\hat{\vx}_{n0}^*\vx_{n0}$, we can write the above inequality as 
\begin{align}\label{eq:initialization-ineq}
\|\sqrt{d_n} \hat{\vx}_{n0} - \beta_0 \vx_{n0}\|_2 \leq 2 \sqrt{d_n} \gamma.
\end{align}
Note that $\vp = \beta_0 \vx_0 \in \setZ$, which follows because $\sqrt{Q} |\beta_0| \|\mC_n \vx_{n0}\|_\infty = |\beta_0| \nu \leq \sqrt{d_n} \nu < 2 \sqrt{d_n} \nu$. Thus, employing $\vp = \beta_0 \vx_0 \in \setZ$ in \eqref{eq:projection-lemma-result}, we obtain $\|\sqrt{d_n} \hat{\vx}_{n0} - \beta_0 \vx_{n0}\|_2 \geq \|\mathbf{v}_{n0} - \beta_0 \vx_{n0}\|_2$. Combining this with \eqref{eq:initialization-ineq} results in
 
\begin{align}\label{eq:initialization-ineq-final}
\|\mathbf{v}_{n0}-\beta_0\vx_{n0}\|_2 \leq 2 \sqrt{d_n} \gamma. 
\end{align}
Similarly, we can show that 
\begin{align}
\|\vu_{n0} - \alpha_{0} \vh_{n0}\|_2 \leq 2 \sqrt{d_n} \gamma,
\end{align}
where $\alpha_{0} = \frac{ \sqrt{d_n}}{d_{n0}} \vh_{n0}^*\hat{\vh}_{n0}$. Finally, 
\begin{align*}
\|\vu_{n0}\mathbf{v}_{n0}^* - \vh_{n0} \vx_{n0}^*\|_F & \leq \|\vu_{n0}\mathbf{v}_{n0}^* - \beta_{0} \vu_{n0} \vx_{n0}^*\| _F  + \|\beta_{0} \vu_{n0} \vx_{n0}^*  - \alpha_{0}\beta_0 \vh_{n0} \vx_{n0}^*\|_F + \|\alpha_{0}\beta_0 \vh_{n0}\vx_{n0}^* - \vh_{n0}\vx_{n0}^*\|_F \\
&\leq \|\vu_{n0}\|_2\|\mathbf{v}_{n0}-\beta_{0}\vx_{n0}\|_2 + |\beta_{0}|\|\vx_{n0}\|_2 \| \vu_{n0} - \alpha_0 \vh_{n0}\|_2 \\
&\qquad + \|(d_n / d_{n0}) \vh_{n0}\vh_{n0}^*\hat{\vh}_{n0}\hat{\vx}_{n0}^* \vx_{n0}\vx_{n0}^*  - \vh_{n0} \vx_{n0}^*\|_F\\
& = \|\vu_{n0}\|_2 \|\mathbf{v}_{n0}-\beta_{0}\vx_{n0}\|_2 +  |\beta_{0}| \| \vu_{n0} - \alpha_0 \vh_{n0}\|_2   + \|d_n\hat{\vh}_{n0}\hat{\vx}_{n0}^* - \vh_{n0} \vx_{n0}^*\|_F \\
& \leq \frac{2}{\sqrt{3}} \cdot 2\sqrt{d_n} \gamma + \sqrt{d_n} \cdot 2 \sqrt{d_n} \gamma + 2d_{n0}\gamma < \frac{20}{3}d_{n0} \gamma,
\end{align*}
which shows that $\|\vu_{n0}\mathbf{v}_{n0}^*-\vh_{n0}\vx_{n0}^*\|_F \leq \frac{2\varepsilon}{5\sqrt{N}\kappa}d_{n0}$ using $\gamma$ defined in \eqref{eq:A*(y)-h0x0-operator-norm}. This shows that $\{\vu_{n0},\mathbf{v}_{n0}\}_{n=1}^N \in \setN_{\frac{2\varepsilon}{5\sqrt{N}\kappa}}$. Plugging $\xi = \tfrac{\varepsilon}{50\sqrt{N}\kappa}$ in \eqref{eq:sample-complexity-initialization} gives the claimed sample complexity and  $1-2\exp\left(-c\xi^2QN/\mu^2\nu^2 \kappa^2 \right)$ gives the probability. 
\end{proof}
\begin{lem}[Theorem 2.8 in \cite{escalante2011alternating}]\label{lem:convex-set-projection} Consider a closed convex set that is not empty is denoted by $\setW$, so we can write
	\begin{align*}
	\Re{\<\vq-\setP_{\setW}(\vq),\vw - \setP_{\setW}(\vq)\>} \leq 0, \ \forall \vw \in \setW, \vq \in \C^Z,
	\end{align*}
	where the projection of $\vq$ onto $\setW$ is represented as $\setP_{\setW}(\vq)$. 
\end{lem}

In this appendix, we have proved the Theorems 1 and 2 stated in our paper.

\bibliographystyle{IEEEtran}
\bibliography{references}

\end{document}